\documentclass[11pt,a4paper]{article}
\pdfoutput=1

\usepackage[margin=1in]{geometry}
\usepackage{amsmath,amssymb,thmtools, amsthm, booktabs,color,doi,graphicx,latexsym,url,xcolor,xspace} 
\usepackage{comment}

\usepackage[numbers,sort&compress]{natbib}
\usepackage{microtype,hyperref}
\usepackage{eucal}
\usepackage[inline]{enumitem}
\setlist[itemize]{label=--}
\setlist[enumerate]{label=(\arabic*),labelindent=\parindent,leftmargin=*}

\definecolor{citecolor}{HTML}{0000C0}
\definecolor{urlcolor}{HTML}{000080}

\usepackage{sectsty}
\allsectionsfont{\boldmath}

\newtheorem{theorem}{Theorem}
\newtheorem{lemma}[theorem]{Lemma}
\newtheorem{corollary}[theorem]{Corollary}
\newtheorem{definition}[theorem]{Definition}

\theoremstyle{remark}
\newtheorem{remark}[theorem]{Remark}

\newcommand{\namedref}[2]{\hyperref[#2]{#1~\ref*{#2}}}
\newcommand{\sectionref}[1]{\namedref{Section}{#1}}
\newcommand{\figureref}[1]{\namedref{Figure}{#1}}
\newcommand{\tableref}[1]{\namedref{Table}{#1}}
\newcommand{\equationref}[1]{\hyperref[#1]{Eq~(\ref*{#1})}}
\newcommand{\theoremref}[1]{\hyperref[#1]{Theorem~\ref*{#1}}}
\newcommand{\lemmaref}[1]{\hyperref[#1]{Lemma~\ref*{#1}}}
\newcommand{\noteref}[1]{\hyperref[#1]{note~\ref*{#1}}}

\newcommand{\F}{\mathcal{F}}

\newcommand{\A}{\mathcal{A}}

\newcommand{\local}{\ensuremath{\mathsf{LOCAL}}\xspace}
\newcommand{\congest}{\ensuremath{\mathsf{CONGEST}}\xspace}

\newcommand{\id}{\ensuremath{\mathsf{ID}}}
\newcommand{\supported}{\emph{supported}\xspace}

\newcommand{\size}[1]{\left| #1 \right|}

\newcommand{\updated}{\dot{E}}

\newcommand{\cc}{\ensuremath{\mathsf{CC}}\xspace}
\newcommand{\rcc}{\ensuremath{\mathsf{RCC}}\xspace}

\DeclareMathOperator{\poly}{poly}

\DeclareMathOperator{\outdeg}{outdeg}

\hypersetup{
    colorlinks=true,
    linkcolor=black,
    citecolor=citecolor,
    filecolor=black,
    urlcolor=urlcolor,
}

\newenvironment{myabstract}
{\list{}{\listparindent 1.5em%
		\itemindent    \listparindent
		\leftmargin    1cm
		\rightmargin   1cm
		\parsep        0pt}%
	\item\relax}
{\endlist}

\newenvironment{mycover}
{\list{}{\listparindent 0pt
		\itemindent    \listparindent
		\leftmargin    1cm
		\rightmargin   1cm
		\parsep        0pt}%
	\raggedright
	\item\relax}
{\endlist}

\newcommand{\myemail}[1]{\,$\cdot$\, {\small #1}}
\newcommand{\myaff}[1]{\,$\cdot$\, {\small #1}\par\medskip}

\hypersetup{
pdftitle={Input-Dynamic Distributed Algorithms for Communication Networks},
pdfauthor={Klaus-Tycho Foerster, Janne H. Korhonen, Ami Paz, Joel Rybicki and Stefan Schmid},
}

\begin{document}

\begin{mycover}
	{\huge\bfseries\boldmath Input-Dynamic Distributed Algorithms for Communication Networks \par}
	\bigskip
	\bigskip
    \textbf{Klaus-Tycho Foerster}
    \myemail{klaus-tycho.foerster@univie.ac.at}
    \myaff{University of Vienna}
    
    \textbf{Janne H.\ Korhonen}
    \myemail{janne.korhonen@ist.ac.at}
    \myaff{IST Austria}

    \textbf{Ami Paz}
    \myemail{ami.paz@univie.ac.at}
    \myaff{University of Vienna}
    
    \textbf{Joel Rybicki}
    \myemail{joel.rybicki@ist.ac.at}
    \myaff{IST Austria}
    
    \textbf{Stefan Schmid}
    \myemail{stefan\_schmid@univie.ac.at}
    \myaff{University of Vienna}
\end{mycover}

\medskip
\begin{myabstract}
  \noindent\textbf{Abstract.}
Consider a distributed task where the \emph{communication network} is fixed but the \emph{local inputs} given to the nodes of the distributed system may change over time. In this work, we explore the following question: if some of the local inputs change, can an existing solution be updated efficiently, in a dynamic and distributed manner?
	
To address this question, we define the \emph{batch dynamic} \congest model in which we are given a bandwidth-limited communication network and a \emph{dynamic} edge labelling defines the problem input.  The task is to maintain a solution to a graph problem on the labeled graph under \emph{batch changes}. We investigate, when a batch of $\alpha$ edge label changes arrive,
\begin{itemize}[nosep]
	\item how much time as a function of $\alpha$ we need to update an existing solution, and
	\item how much information the nodes have to keep in local memory between batches in order to update the solution quickly.
\end{itemize}
Our work lays the foundations for the theory of input-dynamic distributed network algorithms.  We give a general picture of the complexity landscape in this model, design both universal algorithms and algorithms for concrete problems, and present a general framework for lower bounds. In particular, we derive non-trivial upper bounds for two selected, contrasting problems: maintaining a minimum spanning tree and detecting cliques. 
\bigskip
\end{myabstract}



\thispagestyle{empty}
\setcounter{page}{0}
\newpage


\section{Introduction}

There is an ongoing effort in the networking community to render networks more adaptive and ``self-driving'' by automating network management and configuration tasks~\cite{feamster2017and}. To this end, it is essential to design network protocols that solve various tasks, such as routing and traffic engineering, adaptively and \emph{fast}.

Especially in large and long-lived networks, change is inevitable: 
new demand may increase congestion, 
network properties such as link weights may be updated,
nodes may join and leave, and new links may appear or existing links fail.
As distributed systems often need to maintain data structures and other information related to the operation of the network, it is important to update these structures efficiently and reliably upon changes.
Naturally the naive approach of always recomputing everything from scratch after a change occurs might be far from optimal and inefficient. Rather, it is desirable that if there are only \emph{few} changes, the existing solution could be efficiently utilised for computing a new solution.

However, developing robust, general techniques for \emph{dynamic} distributed graph algorithms --- that is, algorithms that reuse and exploit the existence of previous solutions --- is challenging~\cite{awerbuch2008optimal,KonigW13}: even small changes in the communication topology may force communication and updates over long distances or interfere with ongoing updates.
Nevertheless, a large body of prior work has focused on how to operate in dynamic environments where the underlying communication network changes: temporally dynamic graphs~\cite{casteigts2012time} model systems where the communication structure is changing over time; distributed dynamic graph algorithms consider solving individual graph problems when the graph representing the communication networks is changed by addition and removal of nodes and edges~\cite{awerbuch2008optimal,KonigW13,Censor-HillelHK16,GuptaK18,DuZ2018,barenboim2018locally,bamberger2019local,CensorHillelDKPS20}; and self-stabilisation considers recovery from arbitrary transient faults that may corrupt an existing solution~\cite{dijkstra1974self-stabilization,dolev00self-stabilization}.

\subsection{Input-dynamic distributed algorithms}

In contrast to most prior work which focuses on efficiently maintaining solutions in distributed systems where the underlying communication network itself may abruptly change, we instead investigate how to deal with \emph{dynamic inputs} without changes in the topology: we assume that the local inputs (e.g.\ edge weights) of the nodes may change, but the underlying communication network remains static and reliable. We initiate the study of input-dynamic distributed graph algorithms, with the goal of laying the groundwork for a comprehensive theory of this setting. Indeed, we will see that this move from dynamic topology towards a setting more closely resembling centralised dynamic graph algorithms~\cite{Henzinger18}, where input changes and the computational model are similarly  decoupled from each other, enables a development of a general theory of input-dynamic distributed algorithms.

\subsection{Motivation: towards dynamic network management} 

While the input-dynamic distributed setting is of theoretical interest, it is largely motivated by practical questions arising in network management and optimisation.
In wired communication networks the communication topology is typically relatively static (e.g.\ the layout and connections of the physical network equipment), but the input is highly dynamic. For example, network operators perform link weight updates for dynamic traffic engineering \cite{fortz2000internet} or to adjust
link layer forwarding in local area networks~\cite{DBLP:conf/sigcomm/Perlman85,sdn-stp},
content distribution providers dynamically optimise cache assignments~\cite{frank2013pushing},
and traffic patterns naturally evolve over time~\cite{frank2013pushing,sigmetrics20complexity}.
In all these cases, the underlying topology of the communication network remains fixed, but only some input parameters change.

Formally, the above network tasks can often be modelled as algorithmic graph problems, where
the input, in the form of edge weights or communication demands, changes over time, while the network topology remains fixed or changes infrequently.
In the light of the current efforts to render networks more autonomous and adaptive~\cite{DBLP:conf/sigcomm/2018selfdn,juniper-driving,hu-driving}, it is interesting to
understand the power and limitations of such dynamic
distributed optimisations, also compared to a model where the communication topology frequently changes.

To elucidate the connection between fundamental graph problems and basic network management tasks, we discuss some motivating examples.

\paragraph{Example 1: Link layer spanning trees.}
In the context of link layer networking, 
a typical task is to maintain a spanning tree on the network with e.g.\ the (rapid) Spanning Tree Protocol (STP)~\cite{DBLP:conf/sigcomm/Perlman85}. A standard approach to compute spanning trees is that 
the nodes first elect a root (leader) and then pick their parent in the tree according to the shortest distance to the root.
However, when link weights change, the leader is tasked with broadcasting such changes, and can hence take a long time to converge to a new (minimum) spanning tree.
In the case of centralised solutions, such as software-defined networking (SDN)~\cite{sdn-stp}, we run into the same conceptual issues, namely that the changes need to be gathered at a (logically) centralised location, and pushed out to the~network.

From the theoretical perspective,
there are well-established lower bounds for computing a minimum spanning tree \emph{from scratch} in communication-bounded networks: in the classic \congest{} model, the problem requires $\Omega(\sqrt{n} + D)$ rounds, where $n$ is the number of nodes and $D$ is the diameter of the network~\cite{dassarma12}. However, it is not a priori clear how efficiently \emph{maintaining} an MST under input changes can be done in a distributed manner. As one of our results, we show that (1) it is possible to replace the dependency on $\sqrt{n}$ to be linear in the number edge weight changes while maintaining a small local memory footprint \emph{and} (2) there is a matching lower bound showing that this is essentially the best input-dynamic algorithm one can hope for.

\paragraph{Example 2: Traffic engineering and shortest paths routing.}
The reliable and efficient data delivery in ISP networks typically
relies on a clever traffic engineering algorithm~\cite{awduche2002overview}.
Adaptive traffic engineering is usually performed using dynamic
link weight changes, steering the traffic along certain 
(approximately) shortest paths in the network.
The re-computation of routes can be quite time consuming
and exhibit longer convergence times, including the classic  
all-pairs shortest paths (APSP) algorithms 
based on as distance-vector and link-state protocols~\cite{PetersonD11}.
There are theoretical limits on how fast exact and approximate all-pairs shortest paths 
can be computed from scratch in a distributed manner, even when the networks are sparse~\cite{abboud2016near}. 

In this paper we will show that while similar limitations also extend to the case of input-dynamic algorithms, 
there is a simple universal, near-optimal dynamic algorithm for these types of tasks.

\paragraph{Example 3: Detecting network substructures.}
In many networking applications, it is desirable to detect whether the network contains certain substructures and locate them efficiently. For example, the task of cycle detection is a special case of the task of \emph{loop detection}, i.e., detecting whether a routing scheme contains a directed cycle. On the other hand, clique detection can be used for maintaining and assigning e.g.\ failover nodes and links. 

Once again, we can formally characterise the complexity of such tasks, establishing both fast algorithms and proving non-trivial lower bounds 
in the case of input-dynamic network algorithms, 
for several subgraph detection problems.

\paragraph{Outlook: Towards scalable and efficient network management.}
Above we gave three examples how our work can benefit
the study of
standard network management tasks.
In general, we believe that taking the viewpoint of input-dynamic algorithms
has the potential to lead to a paradigm shift in the design of 
efficient and scalable network management protocols.
Current state-of-the-art approaches to dealing with network management
essentially come in two flavours: distributed control with recomputations
from scratch and centralised control (e.g, SDNs). While the former has the 
drawbacks discussed above, limiting the granularity at which optimisations 
can be performed, the latter can entail scalability issues.
For example, the indirection via a controller, even if it is only logically centralised
but physically distributed, can result in delays: in terms of reaction 
and computation time at the controller and in terms of the required synchronisation
in case of multiple controllers (e.g., to keep states consistent)~\cite{DBLP:conf/sds/MichelK17}.

In contrast to the two prior approaches, the paradigm of input-dynamic distributed algorithms aims to realise the best of both worlds: 
design algorithms that 
benefit from previously collected state in distributed protocols,
rapidly generate outputs based on existing solutions, and update the auxiliary data structures for the next set of changes.
%

\subsection{Batch dynamic \textsf{CONGEST} model}

\begin{figure}
  \begin{center}
    \includegraphics[page=5,scale=0.7]{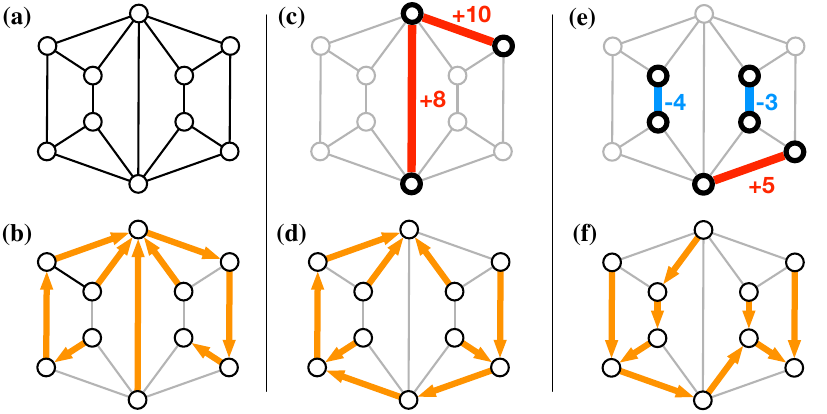}
    \end{center}
  \caption{Examples of input-dynamic minimum-weight spanning tree.  (a) The underlying communication graph, with all edges starting with weight $1$. (b) A feasible minimum-weight spanning tree. (c) A batch of two edge weight increments. (d) Solution to the new input labelling. (e) A new batch of three changes: two decrements and one increment. (f) An updated solution. \label{fig:intro}}
\end{figure}

Our aim is hence to develop a rigorous theoretical framework for reasoning about distributed input-dynamic algorithms. While there are standard models for \emph{non-dynamic} distributed computations~\cite{Peleg00}, there is no established model for \emph{input-dynamic} graph algorithms so far (in \sectionref{sec:related} we overview prior approaches in modelling other aspects of dynamic networks).

To remedy this, we introduce the \emph{batch dynamic} \congest model, which allows us to formally develop a theory of input-dynamic graph algorithms.
In brief, the model is a dynamic variant of the standard \congest model of distributed computation with the following characteristics:
\begin{enumerate}
\item The communication network is represented by a static graph $G = (V,E)$ on $|V|=n$ nodes. The nodes can communicate with each other over the edges, with $O(\log n)$ bandwidth per round.
(This is the standard \congest model~\cite{Peleg00}.)

    \item The input is given by a \emph{dynamic} edge labelling of $G$.
    The input labelling may change and once this happens nodes need to compute a new feasible solution for the new input labelling. The labelling can denote, e.g., edge weights or mark a subgraph of~$G$. 
    We assume that the labels can be encoded using $O(\log n)$ bits so that communicating a label takes a single round.
    \item The goal is to design a distributed algorithm which maintains a solution to a given graph problem on the labelled graph under \emph{batch changes}: up to $\alpha$ labels can change simultaneously, and the nodes should react to these changes.
    The nodes may maintain a \emph{local auxiliary state} to store, e.g., the current output and auxiliary data structures, in order to facilitate efficient updates upon subsequent changes.
\end{enumerate}
\figureref{fig:intro} gives an example of an input-dynamic problem: maintaining a minimum-weight spanning tree.
We define the model in more detail in Section~\ref{sec:model}.

\paragraph{Model discussion.}
As discussed earlier, the underlying motivation for our work is to study how changes to the input can be efficiently handled, while suppressing interferences arising from changing communication topology. A natural starting point for studying communication-efficient solutions for graph problems in networks is the standard \congest model~\cite{Peleg00}, which we extend to model the input-dynamic setting.

Assuming that the communication topology remains static allows us to adopt the basic viewpoint of centralised dynamic algorithms, where an algorithm can fully process changes to input before the arrival of new changes. While this may initially seem restrictive, our algorithms can in fact also tolerate changes arriving during an update: we can simply delay the processing of such changes, and fix them the next time the algorithm starts.
Indeed, this parallels centralised dynamic algorithms, where the processing of changes is not disrupted by a newer change arriving.

While this model has not explicitly been considered in the prior work, we note that input-dynamic distributed algorithms of similar flavour have been studied before in limited manner. In particular, Peleg~\cite{peleg1998distributed} gives an elegant minimum spanning tree update algorithm that, in our language, is a batch dynamic algorithm for minimum spanning tree. Very recently, a series of papers has investigated batch dynamic versions of MPC and $k$-machine models, mainly focusing on the minimum spanning tree problem~\cite{DhulipalaDKPSS20,NowickiO20,GilbertL20}. However, the MPC and $k$-machine models assume a fully-connected communication topology, i.e., every pair of nodes share a direct communication link, making them less suitable for modelling large-scale networks.

Finally, we note that in practice, the communication topology rarely remains static throughout the entire lifetime of a network. However, if the changes in the communication topology are infrequent enough compared to the changes in the inputs, then recomputing a new auxiliary state from scratch, whenever the underlying communication network changes, will have a small cost in the amortised sense. Moreover, any \emph{lower bounds} for the batch dynamic model also hold in the case of networks with changing communication topology.

\section{Contributions}

 \begin{table*}
   \centering
 \caption{Upper and lower bounds for select problems in batch dynamic \congest. Upper bounds marked with $\dagger$ follow from the universal algorithms. The lower bounds apply in a regime where $\alpha$ is sufficiently small compared to $n$, with the threshold usually corresponding to the point where the lower bound matches the complexity of computing a solution from scratch in \congest; see \sectionref{sec:lower} for details. All upper bounds are deterministic. The lower bounds hold for both deterministic and randomised algorithms. \label{tab:results}}
 \vspace{0.5em}

 \begin{tabular*}{0.95\linewidth}{@{}l@{\extracolsep{\fill}}l@{}l@{}l@{}l@{}}
 \toprule
 & \multicolumn{2}{@{}c}{Upper bound} & \multicolumn{1}{@{}c}{Lower bound} &   \\
 \cmidrule{2-3} \cmidrule{4-4}
 Problem & Time & Space & Time & Ref. \\
 \midrule
 any problem & $O(\alpha + D)$ & $O(m \log n)$ & --- & \S\ref{sec:upperbounds} \\
 any $\mathsf{LOCAL}(1)$ problem & $O(\alpha)$ & $O(m \log n)$ & --- & \S\ref{sec:upperbounds} \\
 \midrule
 minimum spanning tree & $O(\alpha + D)$ & $O(\log n)$ & $\Omega(\alpha/\log^2 \alpha + D)$ & \S\ref{sec:MWST}, \S\ref{subsec:lbmst} \\
 $k$-clique & $O(\alpha^{1/2})$ & $O(m \log n)$ & $\Omega(\alpha^{1/4}/\log \alpha)$ & \S\ref{sec:batch-clique}, \S\ref{subsec:lbinst} \\
 \midrule
 $4$-cycle & $O(\alpha)^\dagger$ &$O(m \log n)^\dagger$ & $\Omega(\alpha^{2/3}/\log \alpha)$ & \S\ref{subsec:lbinst} \\
 $k$-cycle, $k \ge 5$ & $O(\alpha)^\dagger$ & $O(m \log n)^\dagger$& $\Omega(\alpha^{1/2}/\log \alpha)$ & \S\ref{subsec:lbinst} \\
 diameter, $(3/2 - \varepsilon)$-apx. & $O(\alpha + D)^\dagger$ & $O(m \log n)^\dagger$ & $\Omega(\alpha/\log^2 \alpha + D)$ & \S\ref{subsec:lbinst} \\
 APSP, $(3/2 - \varepsilon)$-apx. & $O(\alpha + D)^\dagger$ & $O(m \log n)^\dagger$ & $\Omega(\alpha/\log^2 \alpha + D)$ & \S\ref{subsec:lbinst} \\
 \bottomrule
 \end{tabular*}
 \end{table*}

In this work, we focus on the following fundamental questions.
When a batch of $\alpha$ edge label changes arrive,
and the communication graph has diameter $D$, 
\begin{enumerate}[label=(\alph*)]
\item how much time does it take to update an existing solution, as a function of $\alpha$ and $D$, and
\item how much information does a node need to keep in its local memory between batches, in order to achieve optimal running time?
\end{enumerate}
With these questions, we lay the foundations for the theory of input-dynamic distributed graph algorithms. We draw a general picture of the complexity landscape in the batch dynamic \congest model as summarised in \tableref{tab:results}. Our main results are as follows.

\subsection{Universal upper bounds} As an almost trivial baseline, we observe that \emph{any} graph problem can be solved in $O(\alpha + D)$ rounds. Moreover, any graph problem where the output of a node depends only on the constant-radius neighbourhood of the node -- that is, a problem solvable in $O(1)$ rounds in the \textsf{LOCAL} model\footnote{The \textsf{LOCAL} model is similar to the \congest model, but without the $O(\log n)$ limitation on the message sizes~\cite{Peleg00}.} -- can be solved in $O(\alpha)$ rounds. However, these universal algorithms come at a large cost in space complexity: storing the auxiliary state between batches may require up to $O(m \log n)$ bits, where $m$ is the number of edges --- in the input graph if the input marks a subgraph, and in the communication graph if the input represents edge weights. (\sectionref{sec:upperbounds}.)

\subsection{Intermediate complexity: clique enumeration} We give an algorithm for enumerating $k$-cliques in $O(\alpha^{1/2})$ rounds, beating the universal upper bound for local problems, and showing that there exist non-trivial problems that can be solved in $o(\alpha)$ rounds. To complement this result, we show that dynamic clique detection requires $\Omega(\alpha^{1/4})$ rounds. This is an example of a natural problem with time complexity that is neither constant nor $\Theta(\alpha)$. (\sectionref{sec:batch-clique}.)

\subsection{Saving space: minimum-weight spanning trees} We show that a minimum-weight spanning tree can be maintained in $O(\alpha + D)$ rounds using only $O(\log n)$ bits per node for storing the auxiliary state; this exponentially improves the storage requirements of a previous distributed dynamic algorithm of Peleg~\cite{peleg1998distributed}, which uses $O(n \log n)$ bits of memory per node. In addition, we show that our result is tight, in terms of update time, up to $\poly\log\alpha$: for any $\alpha \le n^{1/2}$, maintaining a minimum-weight spanning tree requires $\Omega(\alpha / \log^2 \alpha + D)$ rounds. (\sectionref{sec:MWST}.)

\subsection{A general framework for lower bounds} We develop a framework for lifting \congest lower bounds into the batch dynamic \congest model, providing a vast array of non-trivial lower bounds for input-dynamic problems. These include lower bounds for classic graph problems, such as cycle detection, clique detection, computing the diameter, approximating all-pairs shortest paths, and computing minimum spanning trees.
The lower bounds hold for both deterministic and randomised algorithms.
(\sectionref{sec:lower}.)

\subsection{Dynamic congested clique} We explore the dynamic variant of the \emph{congested clique} model, which arises as a natural special case of the batch dynamic \congest. We show that triangle counting can be solved in $O( (\alpha/n)^{1/3} + 1)$ rounds in this model using $O(n \log n)$ bits of auxiliary state by applying a \emph{dynamic matrix multiplication} algorithm. To contrast this, we show that any problem can be solved in $O(\lceil \alpha / n \rceil)$ rounds using $O(m \log n)$ bits of auxiliary state. (\sectionref{sec:dcm}.)

\subsection{Summary and open questions} As a key takeaway, we have established that the possible time complexities in batch dynamic \congest range from constant to linear-in-$\alpha$, and that there are truly intermediate problems in between. However, plenty of questions remain unanswered; we highlight the following objectives as particularly promising future directions:
\begin{itemize}
    \item \emph{Upper bounds}: Develop new algorithmic techniques for batch dynamic \congest.
    \item \emph{Understanding space}: Develop lower bound techniques for space complexity. In particular, are there problems that exhibit \emph{time-space tradeoffs}, i.e. problems where optimal time and space bounds cannot be achieved at the same time?
    \item \emph{Symmetry-breaking problems}: Understand how problems with subpolynomial complexity in \congest -- in particular, symmetry-breaking problems such as maximal independent set -- behave in the batch dynamic \congest model.
\end{itemize}

\subsection{Technical overview and methodological advancements}

The main \emph{conceptual} contribution of our work is the introduction of the batch dynamic model, which allows the development of a robust complexity theory of input-dynamic distributed algorithms.
A particularly attractive feature of our model is that we can easily leverage standard machinery developed for non-dynamic \congest model in the input-dynamic setting. This for example immediately yields the baseline results given in \sectionref{sec:upperbounds} and the fast triangle counting algorithms for batch dynamic congested clique in \sectionref{sec:dcm}.

However, to obtain \emph{efficient} input-dynamic algorithms, it is necessary to develop new algorithmic and analysis techniques. As our main technical contributions, we analyse two different algorithmic problems, clique enumeration (a local problem) and maintaining minimum-spanning trees (a global problem), and devise a general framework for proving lower bounds for input-dynamic distributed algorithms. 

\paragraph{Clique enumeration.} The clique enumeration problem is a local problem: the nodes need to decide whether their local neighbourhood contains a clique of a certain size. In the dynamic setting, the main challenge is to deal with the fact that nodes do not know the number $\alpha$ of changes in advance, but the running time should be bounded in terms of $\alpha$. When dealing with a global problem that requires $\Omega(D)$ rounds in networks with diameter $D$, we can simply broadcast the number of changes in the network using standard broadcasting techniques.
However, as clique enumeration is a local problem, we wish to obtain running times independent of the diameter of the communication network.

To this end, 
we observe that the subgraph defined by the changed input edges has an useful graph theoretical property, namely, it has bounded degeneracy. This allows us to distributively compute the Nash--Williams decomposition between after a batch of updates, which can be efficiently used to route information about the local changes to input while avoiding congestion.
This resembles to approach taken by e.g.\ Korhonen and Rybicki~\cite{korhonen2017deterministic}, who use a distributed version of the Nash--Williams decomposition by Barenboim and Elkin~\cite{barenboim2010sublogarithmic}, to detect cycles in bounded degeneracy graphs. The main difference to this work is that here we show how to use this approach in the batch dynamic model and we show how to interleave the computation of this decomposition and clique enumeration in a way where nodes only locally have to determine when to halt, without knowing the total number of changes in advance.

\paragraph{Minimum-weight spanning trees.}
Recent work almost exclusively has focused on maintaining minimum-weight spanning trees in fully-connected communication topologies. The main challenge in our work is that we consider general communication topologies, where nodes may need to communicate via large distances. In order to achieve small space complexity, we use a distributed implementation of the standard Eulerian tour tree data structure, which can be used to recover the minimum-weight spanning tree as long as we can maintain the said data structure.

Recently, Eulerian tour trees have also been used in \emph{fully-connected} dynamic models~\cite{italiano2019dynamic,DhulipalaDKPSS20,GilbertL20}, where direct communication is possible between any pair of nodes. In our model, the analysis is complicated by the fact, that communication has to be done over the network e.g.~via broadcast trees -- to avoid congestion, the changes to the input need to be broadcast in a pipelined fashion. The key observation is that the steps required for the Eulerian tour tree update can be formulated as maximum matroid basis problems, which allows us to use the elegant distributed maximum matroid basis algorithm Peleg~\cite{peleg1998distributed} to efficiently compute the required changes to this structure.

\paragraph{Lower bound framework.}
The main technical challenge here is to extend the notion of lower bound family lower bounds into the batch dynamic setting. While the relevant parameter in \congest is the size $n$ of the network, in the batch dynamic model, the time complexity is (mainly) parameterised in terms of the number $\alpha$ of changes. For this reparameterisation, we introduce the notion of \emph{extension properties} and a padding technique that allow us to embed a hard input graph into a larger communication graph.


\section{Related work}\label{sec:related}

As the dynamic aspects of distributed systems have been investigated from numerous different perspectives, giving a comprehensive survey of all prior work is outside the scope of the current work.
Thus, we settle on highlighting the key differences and similarities between the questions studied in related areas and our work.

\subsection{Centralised dynamic graph algorithms}
Before proceeding to the distributed setting, it is worth noting that dynamic graph algorithms in the \emph{centralised setting} have been a major area of research for several years~\cite{Henzinger18}. This area focuses on designing data structures that admit efficient update operations (e.g.\ node/edge additions and removals) and queries on the graph.

Early work in the area investigated how connectivity properties, e.g., connected components and minimum spanning trees, can be maintained~\cite{HenzingerK99,HolmLT01}. Later work has investigated efficient techniques for maintaining other graph structures, such as spanners~\cite{BernsteinFH19}, emulators~\cite{HenzingerKN16}, matchings~\cite{NeimanS16}, maximal independent sets~\cite{assadi2018sublinear,assadi2019sublinear-in-n}; approximate vertex covers, electrical flows and shortest paths~\cite{BhattacharyaHI18,GoranciHP17,DurfeeGGP19}. Recently, conditional hardness results have been established in the centralised setting~\cite{AbboudW14,AnconaHRWW19,HenzingerKNS15}.

Similarly to our work, the input in the centralised setting is dynamic: there is a stream of update operations on the graph and the task is to efficiently provide solutions to graph problems. Naturally, the key distinction is that changes in the centralised setting arrive sequentially and are handled by a single machine. Moreover, in the distributed setting, we can provide unconditional lower bounds for various input-dynamic graph problems, as our proofs rely on communication complexity arguments.

\subsection{Distributed algorithms in changing communication networks}
The challenges posed by dynamic communication networks --- that is, networks where communication links and nodes may appear or be removed --- have been a subject of  ongoing research for decades. Classic works have explored the connection between synchronous static protocols and \emph{fault-prone} asynchronous computation under dynamic changes to communication topology~\cite{awerbuch1992adapting}. Later, it was investigated how to maintain or recompute local~\cite{parter2016local} and global~\cite{elkin2007spanners} graph structures when communication links may appear and disappear or crash. A recent line of work has investigated how to efficiently fix solutions to graph problems under various distributed settings~\cite{KonigW13,Censor-HillelHK16,GuptaK18,DuZ2018,assadi2018sublinear,assadi2019sublinear-in-n,bamberger2019local,CensorHillelDKPS20,DBLP:conf/icdcn/FoersterRSW17,DBLP:journals/tcs/FoersterLSW18,DBLP:conf/nca/Foerster019}. Another line of research has focused on time-varying communication networks which come with temporal guarantees, e.g., that every $T$ consecutive communication graphs share a spanning tree \cite{casteigts2012time,kuhn2010distributed,DBLP:conf/dialm/ODellW05}.

In the above settings, the input graph and the communication network are the same, i.e., the inputs and communication topology are typically coupled. However, there are exceptions to this, as discussed next.

\subsection{Input-dynamic parallel and distributed algorithms}
Several instance of distributed dynamic algorithms can be seen as examples of the input-dynamic approach.
Italiano~\cite{Italiano91} and later Cicerone et al.~\cite{cicerone2003fully}, considered the problem of maintaining a solution all-pairs shortest paths problem when a single edge weight may change at a time. Peleg~\cite{peleg1998distributed} considered the task of correcting a minimum-weight spanning tree after changes to the edge weights, albeit with a large cost in local storage, as the algorithm stores the entire spanning tree locally at each node.

More recently, there has been an increasing interest in developing dynamic graph algorithms for classic parallel models~\cite{acar2011parallelism,simsiri2016work,acar2017parallel,acar2019parallel,tseng2019batch} and massively parallel large-scale systems~\cite{italiano2019dynamic,DhulipalaDKPSS20,NowickiO20,GilbertL20}. In the former, communication is via shared memory, whereas in the latter the communication is via message-passing in a fixed, fully-connected network, but the input is distributed among the nodes and the communication bandwidth (or local storage) of the nodes is limited. Thus, the key difference is that in these parallel models, the communication topology always forms a fully-connected graph, whereas in the batch dynamic \congest considered in our work, the communication topology can be arbitrary, and thus, communication also incurs a distance cost. However, we note that the dynamic congested clique model we study in Section~\ref{sec:dcm} falls under this category.

\subsection{Self-stabilisation} The area of self-stabilisation \cite{dijkstra1974self-stabilization,dolev00self-stabilization} considers robust algorithms that \emph{eventually} recover from \emph{arbitrary} transient failures that may corrupt the state of the system. Thus, unlike in our setting where the auxiliary state and communication network are assumed to be reliable, the key challenge in self-stabilisation is coping with possibly adversarial corruption of local memory and inconsistent local states, instead of changing inputs.

\subsection{Supported models}
Similar in spirit to our model is the \supported \congest model, a variant of the \congest model designed for software-defined networks~\cite{schmid2013exploiting}. In this model, the communication graph is known to all nodes and the task is to solve a graph problem on a given \emph{subgraph}, whose edges are given to the nodes as inputs. The idea is that the knowledge of the communication graph may allow for \emph{preprocessing}, which may potentially offer speedup for computing solutions in the subgraph. However, unlike the batch dynamic \congest model, the supported \congest model focuses on \emph{one-shot} computation. 
Korhonen and Rybicki~\cite{korhonen2017deterministic} studied the complexity of subgraph detection problems in supported \congest. Later, somewhat surprisingly, Foerster et al.~\cite{foerster2019preprocessing} showed that in many cases knowing the communication graph does not help to circumvent \congest lower bounds. Lower bounds were also studied in the supported \local model, for maximum independent set approximation~\cite{supported-local}.


\section{Batch dynamic \textsf{CONGEST} model}\label{sec:model}

In this section, we formally define the the batch dynamic \congest model.

\subsection{Communication graph and computation} The communication graph is an undirected, connected graph $G = (V,E)$ with $n$ nodes and $m$ edges. We use the short-hands $E(G) = E$ and $V(G) = V$.
Each node has a unique identifier of size $O(\log n)$ bits. In all cases, $n$ and $m$ denote the number of vertices and edges, respectively, in $G$, and $D$ denotes the diameter of $G$.

All computation is performed using the graph $G$ for communication, as in the case of the standard \congest model~\cite{Peleg00}: in a single synchronous round, all nodes in lockstep
\begin{enumerate}
\item send messages to their neighbours,
\item wait for messages to arrive, and
\item update their local states.
\end{enumerate}
We assume $O(\log n)$ bandwidth per edge per synchronous communication round. To simplify presentation, we assume that any $O(\log n)$-bit message can be sent in one communication round. Clearly, this only affects constant factors in the running times of the algorithms we obtain.

\subsection{Graph problems} A \emph{graph problem} $\Pi$ is given by sets of input labels $\Sigma$ and output labels $\Gamma$. For each graph $G = (V,E)$, unique ID assignment $\id \colon V \to \{1, \ldots, \poly(n) \}$ for $V$ and input labelling of edges $\ell \colon E \to \Sigma$, the problem $\Pi$ defines a set $\Pi(G,\ell)$ of valid output labellings of form $f \colon V \to \Gamma$. We assume that input labels can be encoded using $O(\log n)$ bits, and that the set $\Pi(G,\ell)$ is finite and computable.
We focus on the following problem categories:
\begin{itemize}
    \item \emph{Subgraph problems:} The input label set is $\Sigma = \{ 0, 1 \}$, and we interpret a labelling as defining a subgraph $H = (V, \{ e \in E \colon \ell(e) =1 \})$. Note that in this case, the diameter of the input graph $H$ can be much larger than the diameter $D$ of the communication graph $G$, but we still want the running times of our algorithms to only depend on $D$.
    \item \emph{Weighted graph problems:} The input label set is $\Sigma = \{ 0, 1, 2, \dotsc, n^C \}$ for a constant $C$, i.e., the labelling defines weights on edges. We can also allow negative weights of absolute value at most $n^C$, or allow some weights to be infinite, denoted by~$\infty$.
\end{itemize}

\subsection{Batch dynamic algorithms}

We define batch dynamic algorithms via the following setting: assume we have some specified input labels $\ell_1$ and have computed a solution for input $\ell_1$. We then change $\alpha$ edge labels on the graph to obtain new inputs $\ell_2$, and want to compute a solution for $\ell_2$. In addition to seeing the local input labellings, each node can store auxiliary information about the previous labelling $\ell_1$ and use it in the computation of the new solution.

More precisely, let $\Pi$ be a problem. Let $\Lambda$ be a set of local auxiliary states; we say that a (global) auxiliary state is a function $x \colon V \to \Lambda$. A \emph{batch dynamic algorithm} is a pair $(\xi, \mathcal{A})$ defined by a set of valid auxiliary states $\xi(G,\ell)$ and a \congest algorithm $\A$ that satisfy the following~conditions:
\begin{itemize}
    \item For any $G$ and $\ell$, the set $\xi(G,\ell)$ is finite and computable. In particular, this implies that there is a (centralised) algorithm that computes some $x \in \xi(G,\ell)$ from $G$ and~$\ell$.
    \item There is a computable function $s \colon \Lambda \to \Gamma$ such that for any $x \in \xi(G,\ell)$, outputting $s(x(v))$ at each node $v \in V$ gives a valid output labelling, that is, $s \circ x \in \Pi(G,\ell)$.
    \item The algorithm $\A$ is a \congest algorithm such that
    \begin{enumerate}[label=(\alph*)]
        \item all nodes $v$ receive as local input the labels on their incident edges in both an old labelling $\ell_1$ and a new labelling $\ell_2$, as well as their own auxiliary state $x_1(v)$ from $x_1 \in \xi(G,\ell_1)$, 
        \item all nodes $v$ will halt in finite number of steps and upon halting produce a new auxiliary state $x_2(v)$ so that together they satisfy $x_2 \in \xi(G,\ell_2)$.
    \end{enumerate}
    Note that we do not require all nodes to halt at the same time. We assume that halted nodes have to announce halting to their neighbours, and will not send or receive any messages after halting. 
\end{itemize}
We define the running time of $\A$ as the maximum number of rounds for all nodes to halt; we use the number of label changes between $\ell_1$ and $\ell_2$ as a parameter and denote this by $\alpha$. The (per node) space complexity of the algorithm is the maximum number of bits needed to encode any auxiliary state $x(v)$ over $x \in \xi(G,\ell)$.

While all of our algorithms are deterministic, one can also consider randomised batch dynamic algorithms. Here one can adopt different correctness and complexity measures. The most common one in centralised and parallel dynamic algorithms (e.g.~\cite{acar2019parallel,tseng2019batch,HenzingerK99,BernsteinFH19}) is to consider \emph{Las Vegas} algorithms. In our setting, this means requiring that, upon halting, nodes always output a valid new auxiliary state; the running time can be measured either (a) by the \emph{expected} running time of the algorithm, or (b) by establishing running time bounds that hold with high probability.
Alternatively, one can also consider \emph{Monte Carlo} algorithms, succeeding with high probability within a fixed running time~(e.g.~\cite{rauch1996improved}); however, these have the disadvantage that the algorithm is likely to fail at some point over an arbitrarily long sequence of batches. Our lower bounds hold for all of these variants, as we discuss in Section~\ref{sec:lower}.

\begin{remark}
Allowing nodes to halt at different times is done for technical reasons, as we do not assume that nodes know the number of changes $\alpha$ and thus we cannot guarantee simultaneous halting in general. Note that with additive $O(D)$ round overhead, we can learn $\alpha$ globally.
\end{remark}

\subsection{Notation}  
Finally, we collect some notation used in the remainder of this paper.
For any set of nodes $U \subseteq V$, we write $G[U] = (U, E')$, where $E'= \{ e \in E : e \subseteq U \}$, for the subgraph of $G$ induced by $U$. For any set of edges $F \subseteq E$, we write $G[F] = (V', F)$, where $V' = \bigcup F$.
When clear from the context, we often resort to a slight abuse of notation and treat a set of edges $F \subseteq E$ interchangeably with the subgraph $(V,F)$ of $G$. Moreover, for any $e = \{u,v\}$ we use the shorthand $e \in G$ to denote $e \in E(G)$.
For any $v \in V$, the set of edges incident to $v$ is denoted by $E(v) = \{ \{u, v\} \in E \}$. The neighbourhood of $v$ is $N^+(v) = \bigcup E(v)$.
We define
\[
 \updated = \{ e \in E : \ell_1(e) \neq \ell_2(e) \}
\]
to be the set of at most $\alpha$ edges whose labels were changed during an update.


\section{Universal upper bounds}
\label{sec:upperbounds}

As a warmup, we establish the following easy baseline result showing that \emph{any} problem $\Pi$ has a dynamic batch algorithm that uses $O(\alpha + D)$ time and $O(m \log n)$ bits of auxiliary space per node: each node simply stores previous input labelling $\ell_1$ as the auxiliary state and broadcasts all changes to determine the new labelling~$\ell_2$.
First, we recall some useful primitives that follow from standard techniques~\cite{Peleg00}.
	
\begin{lemma} \label{lemma:primitives}
  In the \congest model:
 \begin{enumerate}[label=(\alph*)]
 \item A rooted spanning tree $T$ of diameter $D$ of the communication graph $G$ can be constructed in $O(D)$~rounds.

 \item Let $M$ be a set of $O(\log n)$-bit messages, each given to a node. Then all nodes can learn $M$ in $O(|M| + D)$ rounds.
\end{enumerate}
\end{lemma}
With the above routing primitives, it is straightforward to derive the following universal upper bound.

 \begin{theorem}\label{thm:universal-alg}
  For any problem $\Pi$, there exists a dynamic batch algorithm that uses $O(\alpha + D)$ time and $O(m \log n)$ space.
 \end{theorem}

\begin{proof}
  Define $\xi(G, \ell) = \{ \ell \}$, that is, the only valid auxiliary state is a full description of the input. Define the algorithm $\mathcal{A}$ as follows:
  \begin{enumerate}
  \item Let $ \updated  \subseteq E$ be the set of edges that changed. Define
    \[ M = \{ (u, v, \ell_2(\{u,v\})) : \{u,v\} \in  \updated  \}\,. \]
    The set $M$ encodes the $\alpha$ changes and each message in $M$ can be encoded using $O(\log n)$ bits. By \lemmaref{lemma:primitives}b, all nodes can learn the changes in $O(\alpha + D)$ rounds.
  \item Given $M$, each node $v \in V$ can locally construct $\ell_2$ from $M$ and $\ell_1$. Set $x_2(v) = \ell_2$.
  \item Each node $v \in V$ locally computes a solution $s \in \Pi(G, \ell_2)$ and outputs $s(v)$.
  \end{enumerate}
  The claim follows by observing that the update algorithm $\mathcal{A}$ takes $O(\alpha + D)$ rounds and that $\xi(G,\ell) = \{ \ell \}$ can be encoded using $O(m \log n)$ bits.
\end{proof}

 As a second baseline, we consider problems that are strictly local in the sense that there is a constant~$r$ such that the output of a node $v$ only depends on the radius-$r$ neighbourhood of $v$. Equivalently, this means that the problem belongs to the class of problems solvable in $O(1)$~rounds in the \local model, denoted by $\local(1)$.

 \begin{theorem}\label{thm:universal-alg-local}
  For any $\local(1)$ problem, there exists a dynamic batch algorithm that uses $O(\alpha)$ time and $O(m \log n)$ space.
\end{theorem}

\begin{proof}
Let $r$ be the constant such that the output of a node $v$ only depends on the radius-$r$ neighbourhood of $v$. For each node $v$, the auxiliary state is the full description of the input labelling in radius-$r$ neighbourhood of $v$.  Define the algorithm $\mathcal{A}$ as follows:
\begin{enumerate}
    \item Let $ \updated  \subseteq E$ be the set of edges that changed. Define
    \[ M_{v,1} = \{ (u, v, \ell_2(\{u,v\})) : u \in N^+(v), \{u,v\} \in  \updated  \}\,. \]
    The set $M_{v,1}$ encodes the label changes of edges incident to $v$, and each message in $M_{v,1}$ can be encoded using $O(\log n)$ bits.
    \item For phase $i = 1, 2, \dotsc, r$, node $v$ broadcasts $M_{v,i}$ to all of its neighbours, and then announces it is finished with phase $i$. Let $R_{v,i}$ denote the set of messages node $v$ received in phase $i$. Once all neighbours have announced they are finished with phase $i$, node $v$ sets $M_{v,i+1} = R_{v,i} \setminus \bigcup_{j = 1}^i M_{v,j}$ and moves to phase $i+1$.
    \item Once all neighbours of $v$ are finished with phase $r$, node $v$ can locally reconstruct $\ell_2$ in it's radius-$r$ neighbourhood and set the new local auxiliary state $x_2(v)$.
    \item Node $v$ locally computes output $s(v)$ from $x_2(v)$ and halts.
\end{enumerate}
The claim follows by observing that each set $M_{i,v}$ can have size at most $\alpha$, and a node can be in any of the $r = O(1)$ phases for $O(\alpha)$ rounds. In the worst case, the radius-$r$ neighbourhood of a node is the whole graph, in which case encoding the full input labelling takes  $O(m \log n)$ bits.
\end{proof}


\section{Batch dynamic clique enumeration}\label{sec:batch-clique}

In this section, we show that we can do better than the trivial baseline of $O(\alpha)$ rounds for the fundamental local subgraph problem of enumerating cliques.

We consider a setting where the input is a subgraph of the communication graph, represented by  label for each edge indicating its existence in the subgraph.
We show that for any $k \ge 3$, there is a sublinear-time (in~$\alpha$) batch dynamic algorithm for \emph{enumerating} $k$-cliques. More precisely, we give an algorithm that for each node $v$ maintains the \emph{induced} subgraph of its radius-1 neighbourhood. This algorithm runs in $O(\alpha^{1/2})$ rounds and can be used to maintain, at each node, the list all cliques the node is part of. 

To contrast this upper bound, \sectionref{sec:lower} shows that even the easier problem of \emph{detecting} $k$-cliques requires $\Omega(\alpha^{1/4}/\log \alpha)$ rounds.
While this does not settle the complexity of the problem, it shows that this central problem has non-trivial, \emph{intermediate} complexity: more than constant or $\poly\log\alpha$, but still sublinear~in~$\alpha$.

\subsection{Acyclic orientations}\label{subsec:acyclico}

An orientation of a graph $G = (V,E)$ is a map $\sigma$ that assigns a direction to each edge $\{u,v\} \in E$.
For any $d>0$, we say that $\sigma$ is a $d$-orientation if
\begin{enumerate}[noitemsep]
  \item every node $v \in V$ has at most $d$ outgoing edges,
  \item the orientation $\sigma$ is acyclic.
\end{enumerate}
We use $\outdeg_\sigma(v)$ to denote the number of outgoing edges from $v$ in the orientation $\sigma$.

A graph $G$ has \emph{degeneracy} $d$ (``is $d$-degenerate'') if every non-empty subgraph of $G$ contains a node with degree at most $d$.
It is well-known that a graph $G$ admits a $d$-orientation if and only if $G$ has degeneracy of at most~$d$.
We use the following graph theoretic observation.

\begin{lemma}\label{lemma:degeneracy}
  Let $G$ be a $d$-degenerate graph with $n$~nodes and $m$ edges. Then
  $d \le \sqrt{2m}$ and $m \le nd$.
\end{lemma}

\begin{proof}
  For the first claim, suppose that $d > \sqrt{2m}$. Then there is a subset of nodes $U$ such that $G[U]$ has minimum degree $\delta > \sqrt{2m}$. It follows that $U$ has at least $\delta + 1$ nodes, and thus the number of edges incident to nodes in $U$ in $G$ is at least
 \[ \frac{1}{2} \sum_{v \in U} \deg_G(v) \ge \frac{1}{2} \delta (\delta + 1) > \frac{1}{2} \sqrt{2m}(\sqrt{2m} + 1) > m\,,  \]
 which is a contradiction. 
 The second claim follows by considering a $d$-orientation $\sigma$ of $G$ and observing that
   \[
     m = \sum_{v \in V} \outdeg_\sigma(v) \le nd. \qedhere
    \]
\end{proof}

Let $\updated \subseteq E$ be the set of $\alpha$ edges that are changed by the batch update.
We show that the edges of $G[\updated]$ can be quickly oriented so that each node has $O(\sqrt{\alpha})$ outgoing edges despite nodes not knowing~$\alpha$.
This orientation serves as a routing scheme for efficiently distributing relevant changes in the local neighbourhoods.

\begin{lemma}\label{lemma:compute-alpha-orientation}
  An $O(\sqrt{\alpha})$-orientation of $G[\updated]$ can be computed in $O(\log^2 \alpha)$ rounds.
\end{lemma}

\begin{proof}
  Recall that $m$ is the number of edges in the communication graph $G$. Let $H = (U, \updated)$ and note that $|U| \le 2\alpha$ and $\alpha \le m$.
  For an integer $d$, define
  \[
    f(d) = 3 \cdot \sqrt{2^{d+1}} \quad \textrm{ and } \quad T(d) = \left \lceil \log_{3/2} 2^{d+1} \right \rceil.
  \]
  The orientation of $H$ is computed iteratively as follows:
  \begin{enumerate}
  \item Initially, each edge $e \in \updated$ is unoriented.
  \item In iteration $d  =  1, \ldots, \left\lceil \log m \right \rceil $, repeat the following for $T(d)$ rounds:
    \begin{itemize}
    \item If node $v$ has at most $f(d)$ unoriented incident edges, then $v$ orients them outwards and halts. In case of conflict, an edge is oriented towards the node with the higher identifier.
      \item Otherwise, node $v$ does nothing.
  \end{itemize}
\end{enumerate}
Clearly, if node $v$ halts in some iteration $d$, then $v$ will have outdegree at most $f(d)$.
	
Fix $\hat d = \lceil \log \alpha \rceil \le \lceil \log m \rceil$. We argue that by the end of iteration $\hat d$, all edges of $H$ have been oriented. For $0 \le i \le T(\hat d)$, define $U_i \subseteq U$ to be the set of vertices that have unoriented edges after $i \ge 0$ rounds of iteration $\hat d$, i.e.,
\[
U_{i+1} = \{ v \in U_i : \deg_i(v) > f(\hat d) \},
\]
where $\deg_i(v)$ is the degree of node $v$ in subgraph $H_i = H[U_i]$ induced by $U_i$.

Note that every $u \in U \setminus U_0$ has outdegree at most $f(\hat d)$. We now show that each node in $U_0$ halts with outdegree at most $f(\hat d)$ within $T(\hat d)$ rounds.
First, observe that $|U_{i+1}| < \frac{2}{3}|U_i|$. To see why, notice that by \lemmaref{lemma:degeneracy} each $H_i$ has degeneracy at most $\sqrt{2\alpha}$ and thus at most $|U_i| \cdot \sqrt{2\alpha}$ edges. If $|U_{i+1}| \ge \frac{2}{3}|U_i|$ holds, then $H_{i+1}$ has at least
\begin{align*}
  \frac{1}{2} \cdot \sum_{v \in U_{i+1}} \deg_i(v) &> \frac{1}{2} \cdot |U_{i+1}| \cdot f(\hat d) \ge \frac{2}{3} \cdot |U_i| \cdot f(\hat d) \\
  &= |U_i| \cdot \sqrt{2^{\hat d+1}} > |U_i| \cdot \sqrt{2\alpha}
 \end{align*}
edges, which is a contradiction. Thus, we get that $|U_{i+1}| < (2/3)^i \cdot |U|$ and
 \[
   |U_{T(\hat d)}| < (2/3)^{T(\hat d)} \cdot 2 \alpha \le \frac{2 \alpha}{2^{\hat d+1}} \le 1.
\]
Therefore, each edge of $H$ is oriented by the end of iteration $\hat d = \lceil \log \alpha \rceil$ and each node has at most $f(\hat d) = O(\sqrt{\alpha})$ outgoing edges. As a single iteration takes at most $O(\log \alpha)$ rounds, all nodes halt in $O(\log^2 \alpha)$ rounds, as claimed.
\end{proof}

\subsection{Algorithm for clique enumeration}\label{subsec:dynbclique}

Let $G^+[v]$ denote the subgraph induced by the radius-1 neighbourhood of $v$; note that this includes all edges between neighbours of $v$.
Let $H_1 \subseteq G$ and $H_2 \subseteq G$ be the subgraphs given by the previous input labelling $\ell_1$ and the new labelling $\ell_2$, respectively. The auxiliary state $x(v)$ of the batch dynamic algorithm is a map $x(v) = y_v$ such that $y_v : E(G^+[v]) \to \{0,1\}$.
The map $y_v$ encodes which edges in $G^+[v]$ are present in the input subgraph.

The dynamic algorithm computes the new auxiliary state $x_2$ encoding the subgraph $H^+_2[v]$ as follows:
\begin{enumerate}

\item Each node $v$ runs the $O(\alpha^{1/2})$-orientation algorithm on $G[\updated]$ until all nodes in its radius-1 neighbourhood $N^+(v)$ have halted 
  (and oriented their edges in $\updated$).

  \item Let $\updated_\textrm{out}(v) \subseteq \updated$ be the set of outgoing edges of $v$ in the orientation. Node $v \in V$ sends the set
    \[
    A(v) = \{ (e, \ell_2(e)) : e \in \updated_\textrm{out}(v) \}
    \]
    to each of its neighbours $u \in N(v)$.

  \item Define $R(v) = \bigcup_{u \in N^+(v) } A(u)$ and the map $y_v' \colon E(G^+[v]) \to \{0,1\}$ as
    \[
    y'_v(e) = \begin{cases}
      \ell_2(e) & \textrm{if } (e, \ell_2(e)) \in R(v) \\
      y_v(e) & \textrm{otherwise},
    \end{cases}
    \]
    where $y_v$ is the map encoded by the auxiliary state~$x_1(v)$.
  \item Set the new auxiliary state to $x_2(v) = y_v'$.
\end{enumerate}

First, we show that the computed auxiliary state of each node $v$ encodes the subgraph $H_2^+[v]$ induced by the radius-1 neighbourhood of $v$ in the new input graph $H_2$.

\begin{lemma}\label{lemma:clique-encoding-correct}
  Let  $v \in V$ and $e \in G^+[v]$. Then we have $y'_v(e) = 1$ if and only if $e \in H_2^+[v]$.
\end{lemma}

\begin{proof}
  There are two cases to consider. First, suppose $e = \{u,w\} \in \updated$. After Step (1), the edge $\{u,w\}$ is w.l.o.g.\ oriented towards $u$. Hence, in Step (2), if $w \neq v$, then $w$ sends $(e,\ell_2(e)) \in A(w)$ to $v$, as $w \in N(v)$, and if $w=v$ then $v$ knows $A(v)$. Thus, $e \in G^+[v] \cap \updated \subseteq R(v)$. By definition of $y_v'$ it holds that $y_v'(e) = \ell_2(e) = 1$ if and only if $e \in H^+_2[v]$ holds.

  For the second case, suppose $e \notin \updated$. Then, as $H^+_1[v] \setminus \updated = H^+_2[v] \setminus \updated$, and by definition of $y'_v$, we have that $y'_v(e)=y_v(e) = 1$ if and only if $e \in H^+_2[v] \setminus \updated$ holds.
\end{proof}

Next, we upper bound the running time of the above algorithm.

\begin{lemma}
  Each node $v$ computes $H_2^+[v]$ in $O(\alpha^{1/2})$ rounds.
\end{lemma}
\begin{proof}
  By \lemmaref{lemma:compute-alpha-orientation}, Step (1) completes in $O(\log^2 \alpha)$ rounds and $|A| = O(\alpha^{1/2})$. Since each edge in $A$ can be encoded using $O(\log n)$ bits, Step (2) completes in $O(\alpha^{1/2})$ rounds.
  As no communication occurs after Step~(2), the running time is bounded by $O(\alpha^{1/2} + \log^2 \alpha)$. By \lemmaref{lemma:clique-encoding-correct}, node~$v$ learns $H_2^+[v]$ in Step (3).
\end{proof}

Note that if a node $v$ is part of a $k$-clique, then all the edges of this clique are contained in~$H^+_2[v]$. Thus, node $v$ can enumerate all of its $k$-cliques by learning $H^+_2[v]$, and hence, we obtain the following result.

\begin{theorem}
  There exists an algorithm for clique enumeration in the batch dynamic $\congest$ model that runs in $O(\alpha^{1/2})$ rounds and uses $O(m \log n)$ bits of auxiliary state.
\end{theorem}
%

\section{Minimum-weight spanning trees}\label{sec:MWST}

In this section, we construct an algorithm that computes a minimum-weight spanning tree in the dynamic batch model in $O(\alpha + D)$ rounds and using $O(\log n)$ bits of auxiliary state between~batches. For the dynamic minimum spanning tree, we assume that the input label $w(e) \in \{ 0, 1, 2, \dotsc, n^C \} \cup \{ \infty \}$ encodes the weight of edge $e \in E$, where $C$ is a constant, and that the output defines a rooted minimum spanning tree, with each node~$v$ outputting the identifier of their parent.

To do this, we will use a distributed variant of an \emph{Eulerian tour tree}, a data structure familiar from classic centralised dynamic algorithms. In the distributed setting, it allows us to make inferences about the relative positions of edges with regard to the minimum spanning tree \emph{without} full information about the tree. Moreover, the Eulerian tour tree can be compactly encoded into the auxiliary state using only $O(\log n)$ bits per node.

In the following, we first describe a distributed variant of this structure and then how to use it in conjunction with a the \emph{minimum-weight matroid basis} algorithm of Peleg~\cite{peleg1998distributed} to compute the minimum spanning tree in the dynamic batch $\congest$ model efficiently.

\subsection{Distributed Eulerian tour trees}\label{subsec:dist-ett}

We now treat $G = (V,E)$ as a directed graph, where each edge $\{u,v\}$ is replaced with $(u,v)$ and $(v,u)$.
As before, we treat a subgraph $(V,F)$ of $G$ interchangeably with the edge set $F \subseteq E$ and 
further abuse the notation by taking
a subtree $T$ of $G$ to mean a directed subgraph of $G$ containing all directed edges corresponding to an undirected tree on $G$. In particular, we use $|T|$ to denote the number of directed edges in $T$.

Let $H \subseteq G$ be a subgraph of $G$. The bijection $\tau \colon E(H) \to \{ 0, \ldots, \size{E(H)} - 1 \}$ is an \emph{Eulerian tour labelling} of $H$ if the sequence of directed edges $\tau^{-1}(0), \tau^{-1}(1), \ldots, \tau^{-1}( \size{E(H)} - 1 )$ gives an Eulerian tour of $H$.
We say that $u$ is the root of $\tau$ if there is some edge $(u,v)$ such that $\tau(u,v)=0$.
For a map $f \colon A \to B$ and a set $C \subseteq A$, the restriction of $f$ to domain $C$ is the map $f \restriction_C \colon C \to A$ given by $f(c)=f \restriction_C(c)$ for all $c \in C$.

\begin{figure}
\begin{center}
  \includegraphics[page=4,width=0.99\textwidth]{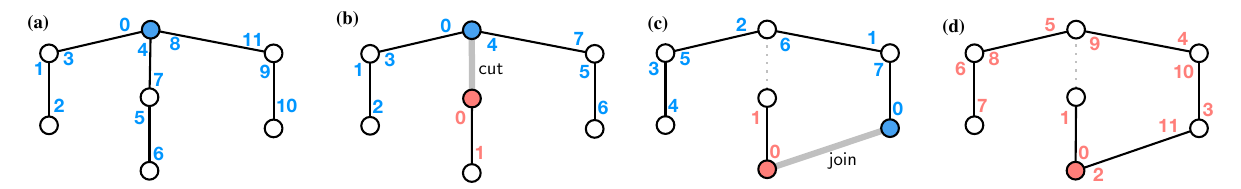}
  \caption{Eulerian tour trees. (a)~An example of Eulerian tour labelling, with root node marked in blue. (b)~The updated Eulerian tour labellings after applying a \textsf{cut} operation. The roots of the new blue and red trees are marked with respective colours. (c) To apply a \textsf{join} operation, we root the red and blue tree to the endpoints of the join edge. (d) Eulerian tour labelling after the \textsf{join} operation. }\label{fig:tours}
\end{center}
\end{figure}

\paragraph{Eulerian tour forests.} An \emph{Eulerian tour forest on $\mathcal{F}$} is a tuple $\mathcal{L} = (L,r,s,a)$ such that

\begin{enumerate}
\item $\mathcal{F} = \{ T_1, \ldots, T_h \}$ is a spanning forest of $G$,
\item $L \colon E \to \{0, \ldots |E|-1\} \cup \{ \infty \}$ is a mapping satisfying the following conditions:
  \begin{itemize}[noitemsep]
    \item for each $T \in \mathcal{F}$ the map $L \restriction_{T}$ is an Eulerian tour labelling of $T$, and
    \item if $(u,v) \notin \bigcup \mathcal{F}$, then $L(u,v) = \infty$,
  \end{itemize}
\item $r \colon V \to V$ is a mapping such that for each $T \in \mathcal{F}$ and node $v \in V(T)$, we have that $r(v)$ is the root of the Eulerian tour labelling $L \restriction_T$ of $T$,
\item $s \colon V \to \{0, \ldots |E|\} $ is a mapping satisfying $s(v) = \size{T}$ for each $T \in \mathcal{F}$ and a node $v \in V(T)$,
\item $a \colon V \to \{0, \ldots |E|-1\}$ is a mapping satisfying for each $T \in \mathcal{F}$ and node $v \in V(T)$ the following conditions:
\begin{itemize}[noitemsep]
    \item if $T$ contains at least one edge, then $a(v) = \min \{ L(e) \colon \text{$e$ is an outgoing edge from $v$} \}$,
    \item if $T$ consists of only node $v$, then $a(v) = 0$.
\end{itemize}
\end{enumerate}
We define distributed operations which allow us to merge any two trees or cut a single tree into two trees, given that all nodes know which edges the operations are applied to. This data structure is then used to efficiently maintain a minimum spanning tree of $G$ under edge weight changes.

\paragraph{Eulerian tour forest operations.}
Let $\mathcal{L}$ be an Eulerian tour forest of $G$.
For any $\mathcal{L} = (L,r,s,a)$ and $E' \subseteq E$, we define the restricted labelling $\mathcal{L}\restriction_{E'} = (L\restriction_{E'}, r \restriction_U, s \restriction_{U}, a \restriction_{U} )$, where $U = \bigcup E'$ is the set of nodes incident to edges in~$E'$.
We implement two operations for manipulating $\mathcal{L}$ (illustrated by \figureref{fig:tours}): a $\mathsf{join}$ operation that merges two trees and a $\mathsf{cut}$ operation that removes an edge from a tree and creates two new disjoint trees. To implement the two basic operations $\mathsf{join}$ and $\mathsf{cut}$, we also use an auxiliary operation $\mathsf{root}$ that is used to reroot a tree.

For brevity, let $T(u)$ to denote the tree node $u$ belongs to in the Eulerian tour forest. We use $\size{T}$ to denote the number of directed edges in $T$. 
The three operations are as follows:

\begin{itemize}
\item $\mathsf{root}(\mathcal{L}, u)$: Node $u$ becomes the root of the tree~$T(u)$.

  \emph{Implementation:} Set
  \begin{align*}
    L(w,v) & \gets L(w,v) - a(u) \bmod s(u) \\
    &\text{for each $(w,v) \in T(u)$, and}\\
    a(v) & \gets a(v) - a(u) \bmod s(u) \\
    &\text{for each $v \in V(T(u))$.}
  \end{align*}
  Otherwise, $L$ and $a$ remain unchanged. Moreover, $r(v) \gets u$ if $r(u)=r(v)$ and otherwise $r$ remains unchanged. All tree sizes remain unchanged.

  \item $\mathsf{join}(\mathcal{L}, e)$: If $e = \{v_i, v_j\}$, where $v_i \in V(T_i)$ and $v_j \in V(T_j)$ for $i \neq j$, then merge $T_i$ and $T_j$ and create an Eulerian tour labelling of $T' = T_i \cup T_j \cup \{e \}$. The root of $T'$ will be the endpoint of $e$ with the smaller identifier.

    \emph{Implementation:} Let $e = \{v_i,v_j\}$, where $v_i \in V(T_i)$ and $v_j \in V(T_j)$ for $i \neq j$. Without loss of generality, suppose $v_i < v_j$. The operation is implemented by the following steps:
    \begin{enumerate}
      \item Run $\mathsf{root}(v_i)$ and $\mathsf{root}(v_j)$.
      \item Set $L(v_i, v_j) \gets s(v_i)$ and $L(v_j, v_i) \gets s(v_i) + s(v_j) + 1$.\\
            For each $(u,v) \in T_j$, set $L(u,v) \gets L(u,v) + s(v_i) + 1$.
      \item For each $u \in V(T_j)$, set $a(u) \gets a(u) + s(v_i) + 1$.
      \item For each $u \in T'$, set $s(u) \gets s(v_i) + s(v_j) + 2$ and $r(u) \gets v_i$.
    \end{enumerate}

  \item $\mathsf{cut}(\mathcal{L}, e)$: For an edge $e = \{v_1,v_2\}$ in some tree $T$, create two new disjoint trees $T_1$ and $T_2$ with Eulerian tour labellings rooted at $v_1$ and $v_2$ such that $T_1 \cup T_2 = T \setminus \{e \}$.

    \emph{Implementation:} Let $e = \{v_1,v_2\}$. Without loss of generality, assume that $a(v_1) < a(v_2)$. Let $z_1 = L(v_1,v_2)$, $z_2 = L(v_2,v_1)$, and $x = z_2-z_1$. The edge labels are updated as follows:
    \begin{enumerate}
     \item Set $L(v_1,v_2) \gets \infty$ and $L(v_2,v_1) \gets \infty$.
     \item  If $a(v) \in [0, \ldots, z_1]$, then set $s(v) \gets s(v) - x - 1$ and $a(v) \gets a(v)$.\\
            If $a(v) \in (z_1, \ldots, z_2]$, then set $s(v) \gets x - 1$ and $a(v) \gets a(v) - z_1 -1$. \\
            Otherwise, set $s(v) \gets s(v) - x - 1$ and $a(v) \gets a(v) - x - 1$.
     \item If $L(u,v) \in [0, \ldots, z_1)$, then set $L(u,v) \gets L(u,v)$.\\
           If $L(u,v) \in (z_1, \ldots, z_2)$, then set $L(u,v) \gets L(u,v) - z_1 - 1$. \\
           Otherwise, set $L(u,v) \gets L(u,v) - x - 1$.
     \item Run $\mathsf{root}(v_1)$ and $\mathsf{root}(v_2)$.
    \end{enumerate}
\end{itemize}

The next lemma shows that the above operations result in a new Eulerian tour forest, i.e., the operations are correct.

\begin{lemma}
  Given an Eulerian tour forest $\mathcal{L}$, each of the above three operations produce a new Eulerian tour forest $\mathcal{L}'$.
\end{lemma}
\begin{proof}
  For the $\mathsf{root}(u)$ operation, observe that only labels in the subtree $T(u)$ change by being shifted by $a(u)$ modulo $|T(u)|$. Hence, the updated labelling of $T(u)$ remains an Eulerian tour labelling. Since the smallest outgoing edge of $u$ will have label $a(u) - a(u) \mod |T(u)| = 0$, node $u$ will be the root of $T(u)$ in the new Eulerian tour labelling of $T(u)$.

  For the $\mathsf{join}(e)$ operation, we observe that after the first step, $v_i$ and $v_j$ are the roots of their respective trees. In particular, after the $\mathsf{root}$ operations, the largest incoming edge of $v_i$ will have label $|T_i|-1$ and the smallest outgoing edge of $v_i$ will have label 0. Hence $v_i$ becomes the root of $T'$. Moreover, in the new Eulerian labelling any edge in $T(v_i)$ will have a valid Eulerian tour labelling, as the labels for $T_i$ remain unchanged. In $T_j$ the labels are a valid Eulerian tour labelling shifted by $|T_i| + 1$. As, in the new labeling the the edge $(v_i,v_j)$ will have label $|T_i|$ and the smallest outgoing label of $v_j$ will be $|T_i|+1$, and the largest incoming label will of $v_j$ will be $|T_i|+|T_j|$. The label of $(v_j,v_i)$ will therefore be $|T_i|+|T_j|+1 = |T_i \cup T_j \cup \{e\}| -1$ and this is the largest label of the new Eulerian tour labelling. Hence, the new labelling is an Eulerian tour forest.

  Finally, consider the $\mathsf{cut}(e)$ operation. Let $T_1$ and $T_2$ be the trees created by removing the edge $e$ from $T$. Note that $x = z_2 - z_1 = |T_2|+1$ and $|T| = |T_1| + |T_2| + 2$, since we are counting directed edges.
  Clearly, after cutting the edge $e$ from $T$, the a node $v$ belongs to subtree $T_2$ if $a(v) \in [z_1, \ldots, z_2)$ and otherwise to $T_1$. Thus, in the latter case $s(v)$ is set to $|T| - x - 1 = |T| - |T_2| - 2 = |T_1|$, and in the former, $s(v)$ is set to $x - 1 = |T_2|$. 

  Suppose an edge $L(u,v) < z_1$. Then the edge $(u,v)$ belongs to $T_1$ and its label will remain unchanged. Now suppose $L(u,v) > z_2$. Then $(u,v)$ will be part of $T_1$ and its new label will be $L(u,v) - x - 1 = L(u,v) - |T_2| - 1$. In particular, the edge of $u_1$ with the smallest outgoing label $z_2+1$ will have the label $z_1$ in the new labeling. Thus, $L$ restricted to $T_1$ will be a valid Eulerian tree tour labelling of $T_1$. It remains to consider the case that $L(u,v) \in (z_1,z_2)$. However, it is easy to check that now the root of $T_2$ will be $v_2$ and the new labelling restricted to $T_2$ will be a valid Eulerian tree tour labelling of $T_2$. Finally, the $\mathsf{root}$ operations ensure that the endpoints of $e$ become the respective roots of the two trees, updating the variables $r(\cdot)$.
\end{proof}

A key property of the Eulerian tour forest structure is that any node that knows the labels of a set $E' \subseteq E$ can locally deduce the new labels of all edges in $E'$ after either $\mathsf{join}$ or $\mathsf{cut}$ operation is applied to a given edge in~$E'$.

\begin{lemma}\label{lemma:local-updates}
   Let $\mathcal{L}$ be an Eulerian tour forest and $e,f \in E$ be edges.
   Suppose $\mathcal{L}'$ is obtained by applying either the $\mathsf{join}(\mathcal{L}, e)$ or the  $\mathsf{cut}(\mathcal{L},e)$ operation.
   Then $\mathcal{L'} \restriction_{\{e,f\}}$ can be computed from $\mathcal{L} \restriction_{\{e,f\}}$. 
\end{lemma}

\begin{proof}
  Let $e = \{ u_1, u_2\}$ and $f = \{v_1, v_2\}$. Let $f \neq e$ be an
	edge whose labels we need to compute after an operation on $e$. We show that after applying any one of the three operations on $\mathcal{L}$, the labels $\mathcal{L}'\restriction_{f}$ can be computed from $\mathcal{L} \restriction_{\{e \cup f\}}$. There are three cases to consider:
  \begin{enumerate}
  \item $\mathcal{L}' = \mathsf{root}(\mathcal{L}, u)$: If $f \notin T(u)$, then $\mathcal{L} \restriction_f = \mathcal{L}' \restriction_f$, as the labels of $f$ do not change. If $f \in T(u)$, then $\mathcal{L}'\restriction_f$ depends only on $a(u)$ and $s(u)$.

  \item $\mathcal{L}' = \mathsf{join}(\mathcal{L}, e)$: If $f \notin T_1 \cup T_2$, then
    $\mathcal{L} \restriction_f = \mathcal{L}' \restriction_f$, as the labels of $f$ do not change after the joining these two trees. Hence suppose $f \in T_1 \cup T_2$. From the previous case, we know that the two $\mathsf{root}$ operations depend on $a(u_i)$ and $s(u_i)$ for $i \in \{1,2\}$. The latter two steps depend only on $s(u_i)$. As these values are contained in $\mathcal{L} \restriction_{\{e,f\}}$, the restriction $\mathcal{L}' \restriction_{f}$ is a function of $\mathcal{L} \restriction_{\{e,f\}}$.

  \item $\mathcal{L}' = \mathsf{cut}(\mathcal{L}, e)$: If $f \notin T$, then the labels of $f$ do not change. Hence, suppose $f \in T$. One readily checks that the update operations in Steps 1-3 depend on $z_1 = L(u_1,u_2)$, $z_2 = L(u_2,u_1)$, $a(v_i)$ and $s(v_i)$ for $i \in \{1,2\}$. Therefore,  $\mathcal{L}' \restriction_{f}$ is a function of $\mathcal{L} \restriction_{ \{e,f\}}$.
  \end{enumerate}
  Thus, in all cases  $\mathcal{L}' \restriction_{f}$ is a function of $\mathcal{L} \restriction_{ \{e,f\}}$, and the claim follows.
\end{proof}

\paragraph{Storing the Eulerian tour tree of a minimum-weight spanning tree.}
Suppose $\mathcal{L}$ is an Eulerian tour forest  on the minimum-weight spanning tree of $G$.
Later, our algorithm will in fact always maintain such a Eulerian tour forest after a batch of updates.

The auxiliary state $x$ is defined as follows. For each node $v \in V$, the auxiliary state $x(v)$ consists of the tuple $(r(v), p(v), \lambda(v))$, where
\begin{itemize}[noitemsep]
  \item $r(v)$ is the identifier of the root of the spanning tree,
  \item $p(v)$ points to the parent of $v$ in the spanning tree,
  \item $\lambda(v) = \big( L(p(v),v), L(v,p(v)) \big)$, respectively.
\end{itemize}
These variables can be encoded in $O(\log n)$ bits. Moreover, each node $v$ can reconstruct $\mathcal{L} \restriction_{E(v)}$ from the auxiliary state $x$ in $O(1)$ rounds. 
\begin{lemma}\label{lemma:mst-auxiliary}
  Given the auxiliary state $x_1(v)$ that encodes $\mathcal{L}$ on a spanning tree of $G$, each node $v$ can learn in $O(1)$ communication rounds $\mathcal{L} \restriction_{E(v)}$. Likewise, given $\mathcal{L}\restriction_{E(v)}$, node $v$ can compute the corresponding auxiliary state $x_1(v)$ locally.
\end{lemma}

\begin{proof}
Since $\mathcal{L}$ is an Eulerian tour forest on a spanning tree, every node $v$ knows $s$ (the size of the spanning tree) and $r$ (the root of the tree), as both are constant functions. As $\lambda(v)$ can be encoded using $O(\log n)$ bits, each node $v$ can send $\lambda(v)$ to all of its neighbours in $O(1)$ communication rounds. Thus, after $O(1)$ rounds node $v$ knows $L\restriction_{E(v)}$. The second part follows directly from the definition of $x_1(v)$.
\end{proof}

\subsection{Maximum matroid basis algorithm}\label{subsec:maxmatroid}

We use an algorithm of Peleg~\cite{peleg1998distributed} as a subroutine for finding minimum and maximum weight \emph{matroid bases} in distributed manner. We first recall the definition of matroids.

\begin{definition}
A \emph{matroid} is a pair $\mathcal{M} = (A, \mathcal{I})$, where $A$ is a set and $\mathcal{I} \subseteq 2^A$ satisfies the following:
\begin{enumerate}[noitemsep]
    \item The family $\mathcal{I}$ is non-empty and closed under taking subsets.
    \item For any $I_1, I_2 \in \mathcal{I}$, if $\size{I_1} > \size{I_2}$, then there is an element $x \in I_1 \setminus I_2$ such that $I_2 \cup \{ x \} \in \mathcal{I}$. This is called the \emph{augmentation property} of a matroid.
\end{enumerate}
We say that a set $I \subseteq A$ is \emph{independent} if $I \in \mathcal{I}$. A maximal independent set is called a \emph{basis}.
\end{definition}

In the \emph{maximum matroid basis problem}, we are given a matroid $\mathcal{M} = (A, \mathcal{I})$ with a weight function $w \colon A \to \{ -n^C, \dotsc, n^C \}$ giving unique weights for all elements, and the task is to a find a basis $B$ of $\mathcal{M}$ with maximum weight $w(B) = \sum_{x \in B} w(x)$. In more detail, the input is specified as follows:
\begin{itemize}
    \item Each node receives a set $A_v \subseteq A$ as input, along with the associated weights. We have a guarantee that $\bigcup_{v \in V} A_v = A$, and the sets $A_v$ may overlap.
    \item Each element $x \in A$ is decorated with additional data $M(x)$ of $O(\log n)$ bits, and given $M(A')$ for $A' \subseteq A$, a node $v$ can locally compute if $A'$ is independent in $\mathcal{M}$.
\end{itemize}
As output, all nodes should learn the maximum-weight basis $B$. Note that since negative weights are allowed and all bases have the same size, this is equivalent to finding a \emph{minimum-weight} matroid basis. 

\begin{theorem}[\citep{peleg1998distributed}]\label{thm:matroid-basis}
The distributed maximum matroid basis problem over $\mathcal{M}$ can be solved in $O(\alpha + D)$ rounds, where $\alpha$ is the size of bases of $\mathcal{M}$.
\end{theorem}

\subsection{Cycle and cut properties}
We make use of the following well-known cycle and cut properties of spanning trees.

\begin{lemma}\label{lemma:cycle-and-cut-properties}
  Suppose the weights of the graph $G$ are unique. Then the following hold:
  \begin{itemize}
    \item Cycle property: For any cycle $C$, the heaviest edge of $C$ is not in minimum-weight spanning tree~of~$G$.

    \item Cut property: For any set $X \subseteq V$, the lightest edge between $X$ and $V \setminus X$ is in the minimum-weight spanning tree of $G$.
  \end{itemize}
\end{lemma}

\subsection{Maintaining a minimum spanning tree}\label{subsec:maintaintree}

Let $G_1 = (V,E, w_1)$ and $G_2 = (V,E,w_2)$ be the graph before and after the $\alpha$ edge weight changes.
Since each edge is uniquely labelled with the identifiers of the end points, we can define a global total order on all the edge weights, where edges are ordered by weight and any equal-weight edges are ordered by the edge identifiers.
Let $T^*_1$ and $T^*_2$ be the unique minimum-weight spanning trees of $G_1$ and $G_2$, respectively.

\paragraph{Communicated messages.} We now assume that each \emph{communicated} edge $e$ is decorated with the tuple
  $M(e) = \left( \mathcal{L}\restriction_{\{e\}}, w_1(e), w_2(e) \right)$.
For a set $E$, we write $M(E) = \{ M(e) \colon e \in E \}$.
  Note that for any edges $e,e' \in E$, the information $M(e)$ and $M(e')$ suffice to compute $\mathcal{L}' \restriction_{e'}$ after either a $\mathsf{join}(\mathcal{L},e)$ or $\mathsf{cut}(\mathcal{L},e)$ operation on $\mathcal{L}$, by \lemmaref{lemma:local-updates}.
  Since $M(e)$ can be encoded in $O(\log n)$ bits, the message encoding $M(e)$ can be communicated via an edge in $O(1)$ rounds.

\paragraph{Overview of the algorithm.}
The algorithm heavily relies on using a BFS tree $\mathcal{B}$ of the communication graph $G$ as a broadcast tree, given by \lemmaref{lemma:primitives}. Without loss of generality, observe that we can first process at most $\alpha$ weight \emph{increments} and then up to $\alpha$ weight \emph{decrements} afterwards.
On a high-level, the algorithm is as follows:
\begin{enumerate}
\item Let $E^+ = \{ e : w_2(e) > w_1(e) \}$ and $E^- = \{ e : w_2(e) < w_1(e) \}$.


\item Solve the problem on the graph $G_1'$ obtained from $G_1$ by changing only the weights in $E^+$.

\item Solve the problem on the graph $G_2$ obtained from $G_1'$ by changing the weights in $E^-$.

\end{enumerate}
We show that Steps (2)--(3) can be done in $O(\alpha + D)$ rounds, which yields the following result.

\begin{theorem}\label{thm:matroid-basis-algorithm}
There is an algorithm for minimum-weight spanning trees in the batch dynamic \congest model that runs in $O(\alpha + D)$ rounds and uses $O(\log n)$ bits per node to store the auxiliary state.
\end{theorem}

\subsection{Handling weight increments}

We now design an algorithm that works in the case $|E^+| \le \alpha$ and $E^- = \emptyset$. That is,
the new input graph $G_2$ differs from $G_1$ by having only the weights of edges in $E^+$ incremented.
Let $T^*_1$ and $T^*_2$ be the minimum spanning trees of $G_1$ and $G_2$, respectively.
Note that $F = T^*_1 \setminus E^+$ is a forest on $G_1$ and $G_2$ and $w_1(F)=w_2(F)$, splitting the graph into connected components.
Let $A^* \subseteq E \setminus F$ be the \emph{lightest} set of edges connecting the components of $F$ under weights $w_2$.

\begin{lemma}\label{lemma:increment-fix}
  The spanning tree $F \cup A^*$ is the minimum-weight spanning tree of $G_2$.
\end{lemma}

\begin{proof}
  Suppose there exists some edge $e \in F \setminus T^*_2$. Let $u$ be a node incident to $e$ and let $S \subseteq V$ be the set of nodes in the connected component of $u$ in $F \setminus \{ e \}$. By the \emph{cut property} given in \lemmaref{lemma:cycle-and-cut-properties}, the lightest edge $f$ (with respect to $w_2$) in the cut
 between $S$ and $V \setminus S$ is in the minimum spanning tree $T^*_2$.
  Since $e \notin T^*_2$ and $f \in T^*_2$, we have that $w_2(f) < w_2(e)$. By definition, $e \in F$ implies that $e \notin E^+$, and hence,
  \[
  w_1(f) \le w_2(f) < w_2(e) = w_1(e).
  \]
  Thus, there exists a spanning tree $T' = (T^*_1 \setminus \{ e \}) \cup \{ f \}$ such that $w_1(T') < w_1(T^*_1)$. But by definition of $F$, we have $e \in F \subseteq T^*_1$, which is a contradiction. Hence, $F \subseteq T^*_2$.
  Since $F \subseteq T^*_2$ is a forest and $A^*$ is the lightest set of edges that connects the components of $F$, the claim follows.
\end{proof}

We show that the set $A^*$ can be obtained as a solution to a minimum matroid basis problem, and thus can be computed in $O(\alpha + D)$ communication rounds. 
In the following, we assume that the auxiliary state encodes an Eulerian tour forest $\mathcal{L}$ on $T^*_1$. We first show that $A^*$ is a minimum-weight basis of an appropriately chosen matroid. Let $A$ be the set of \emph{all} edges that connect components of $F$ in $G_2$.

\begin{lemma}\label{lemma:component-mst-matroid}
  Let $\mathcal{I} = \{ I \subseteq A \colon F \cup I \text{ is acyclic on $G_2$} \}$. Then $\mathcal{M} = (A, \mathcal{I})$ is a matroid and the minimum-weight basis of $\mathcal{M}$ is $A^*$.
\end{lemma}

\begin{proof}
We note that $\mathcal{M}$ is matroid, as it's the contraction of the graphical matroid on $G$ (see e.g.~\cite[Part IV: Matroids and Submodular Functions]{schrijver-book}%
). Moreover, for any basis $B \in \mathcal{I}$, the set $F \cup B$ is a spanning tree on $G_2$ with weight $w_2(F) + w_2(B)$. Since $A^* \in \mathcal{I}$ and $F \cup A^*$ is the unique minimum spanning tree on $G_2$, it follows that $A^*$ is the minimum-weight basis for $\mathcal{M}$.
\end{proof}

To apply the minimum matroid basis algorithm of Theorem~\ref{thm:matroid-basis-algorithm}, we next show that nodes can locally compute whether a set is independent in the matroid $\mathcal{I}$, given appropriate information.

\begin{lemma} 
  Assume a node $v$ knows $M(E^+)$ and $M(X)$ for a set $X \subseteq A$.
  Then $v$ can locally determine if $X$ is independent in $\mathcal{M}$. \label{lemma:v-local-independent}
\end{lemma}

\begin{proof}
Recall that $\mathcal{L}$ is the fixed Eulerian tour forest on $T_1^*$ encoded by the auxiliary data of the nodes and messages $M(e)$. By definition, node $v$ can obtain $\mathcal{L} \restriction_{ E^+ \cup X}$ from $M(E^+)$ and $M(X)$. Let $X = \{ e_1, e_2, \dotsc, e_k \}$. To check that $X$ is independent, i.e.\ $F \cup X$ is a forest, node $v$ uses the following procedure:
   \begin{enumerate}
     \item Let $\mathcal{L}_0\restriction_{E^+ \cup X}$ be the Eulerian tour forest on $F$ obtained from $\mathcal{L} \restriction_{E^+ \cup X}$ by applying the $\mathsf{cut}$ operation for each $e \in E^+$ in sequence.
     \item For $i \in \{1, \ldots, k \}$ do the following:
       \begin{enumerate}[label=(\alph*)]
       \item Determine from $\mathcal{L}_{i-1}\restriction_{e_i}$ if the endpoints $u$ and $w$ of $e_i$ have the same root, i.e. $r(u) = r(w)$. If this is the case, then $F \cup \{ e_1, e_2, \dotsc, e_i \}$ has a cycle, and node $v$ outputs that $X$ is not independent and halts. 
       \item Compute $\mathcal{L}_i \restriction_X = \mathsf{join}(\mathcal{L}_{i-1}  \restriction_X, e_i)$.
       \end{enumerate}
    \item Output that $X$ is independent.
   \end{enumerate}
If $X$ is not independent, then $F \cup X$ has a cycle and algorithm will terminate in Step 2(a). Otherwise, $F \cup X$ is a forest, and the algorithm will output that $X$ is independent.
\end{proof}

\paragraph{Algorithm for handling weight increments.}
The algorithm for maintaining minimum spanning trees under weight increments is now as follows:
\begin{enumerate}
\item Each node $v$ computes its local Euler tour forest labelling $\mathcal{L} \restriction_{E(v)}$ from the auxiliary state $x_1(v)$.

\item Broadcast $M(e)$ for each $e \in E^+$ using the broadcast tree $\mathcal{B}$ given by \lemmaref{lemma:primitives}.
  
\item Use the minimum matroid basis algorithm over $\mathcal{M}$ to compute $A^*$.

\item Each node $v$ locally computes $\mathcal{L}_1 \restriction_{E(v)}$ by applying the $\mathsf{cut}$ operation on each edge in $E^+ \setminus A^*$ in lexicographical order, starting from $\mathcal{L} \restriction_{E(v)}$.
\item Each node $v$ locally computes $\mathcal{L}_2 \restriction_{E(v)}$ by applying the $\mathsf{join}$ operation on each edge in $A^* \setminus E^+$ in lexicographical order, starting from $\mathcal{L}_1 \restriction_{E(v)}$.
\item Each node $v$ outputs local auxiliary state $x_2(v)$ corresponding to $\mathcal{L}_2$.
\end{enumerate}

\begin{lemma}
The above algorithm solves batch dynamic minimum-weight spanning trees under edge weight increments in $O(\alpha + D)$ rounds.
\end{lemma}

\begin{proof}
By Lemma~\ref{lemma:mst-auxiliary}, Step (1)~of the algorithm can be done in $O(1)$ rounds, and by Lemma~\ref{lemma:primitives}, Step~(2) can be done in $O(\alpha + D)$ rounds. Step~(3) can be implemented in $O(\alpha + D)$ rounds by Theorem~\ref{thm:matroid-basis} and Lemma~\ref{lemma:component-mst-matroid}, and after Step~(3) all nodes have learned the set $A^*$. Since all nodes apply the same operations to the Eulerian tour forest in the same order in Steps~(4) and (5), all nodes produce compatible auxiliary states in Step~(6).
\end{proof}

\subsection{Handling weight decrements}

We now consider the dual case, where $|E^-| \le \alpha$ and $E^+ = \emptyset$.
Let $B = T^*_1 \cup E^-$ and $\mathcal{C}$ be the set of cycles in $B$.
Let $B^* \subseteq B$ be the heaviest edge set such that $B \setminus B^*$ is a spanning tree.

\begin{lemma}
  The spanning tree $B \setminus B^*$ is the minimum spanning tree of $G_2$.
  \label{lemma:st-b-star}
\end{lemma}
\begin{proof}
  Let $e \in T^*_2$ and suppose $e \notin B = T^*_1 \cup E^-$. Since  $e \notin T^*_1$ the edge $e$ creates a unique cycle $C$ in $T^*_1$. The edge $e$ is the heaviest edge on cycle $C$ under weights $w_1$, as otherwise we would obtain a spanning tree lighter than $T^*_1$ by replacing the heaviest edge on $C$ by $e$. Since we assume no weight increments and $e \notin E^-$, edge $e$ remains the heaviest edge on the cycle $C$ also under the new edge weights $w_2$. Hence, $e \notin T^*_2$ by the cycle property, which contradicts our initial assumption. Thus, $T^*_2 \subseteq B$.

Now consider any spanning tree $T \subseteq B$. All spanning trees have the same number of edges, and we have $w_2(T) = w_2(B) - w_2(B \setminus T)$. Thus, for the minimum spanning tree $T$ the weight $w_2(B \setminus T)$ is maximised. Since the complement of any spanning tree cuts all cycles in $B$, we have $T^*_2 = B \setminus B^*$.
\end{proof}

\begin{lemma}
Let
\[ \mathcal{J} = \{ J \subseteq B \colon B \setminus J  \text{ contains a spanning tree of $B$} \}\,.\]
Then $\mathcal{N} = (B, \mathcal{J})$ is a matroid and the maximum-weight basis of $\mathcal{N}$ is $B^*$.
\label{lemma:j-matroid-basis}
\end{lemma}
\begin{proof}
We have that $\mathcal{N}$ is the dual of the graphical matroid on $(V, B)$, and thus a matroid (see e.g.~\cite[Part IV: Matroids and Submodular Functions]{schrijver-book}%
). Moreover, $B^*$ is the complement of the minimum spanning tree and thus maximum-weight basis of $\mathcal{N}$.
\end{proof}

\begin{lemma}\label{lemma:mst-dual-matroid}
  Assume a node $v$ knows $M(E^-)$ and $M(X)$ for a set $X \subseteq B$. Then $v$ can locally determine if $X$ is independent in $\mathcal{N}$.
\end{lemma}
\begin{proof}
We observe that $X$ is independent in $\mathcal{N}$ if and only if the edge set $B \setminus X$ spans the graph $G_2$, directly by definitions. Thus, we implement the independence check by using local Eulerian tour forest operations to check if we can obtain a spanning tree $T \subseteq B \setminus X$, by starting from the old minimum spanning tree $T_1^*$, deleting all edges from $X$, and then adding edges from $E^-$ to complete the tree if possible.
    
In more detail, the algorithm works as follows. Recall that by definition, node $v$ can compute $\mathcal{L}\restriction_{E^- \cup X}$ from $M(E^-)$ and $M(X)$. Let $E^- \setminus X = \{ e_1, e_2, \dotsc, e_k \}$.
\begin{enumerate}
     \item Let $\mathcal{L}_0\restriction_{E^- \cup X}$ be the Eulerian tour forest on $B$ obtained from $\mathcal{L} \restriction_{E^- \cup X}$ by applying the $\mathsf{cut}$ operation for each $e \in X \cap T_1^*$ in sequence. Note that node can check directly from $\mathcal{L}\restriction_{E^- \cup X}$ which edges in $X$ are in the minimum spanning tree $T_1^*$.
     \item For $i \in \{1, \ldots, k \}$ do the following:
       \begin{enumerate}[label=(\alph*)]
       \item Determine from $\mathcal{L}_{i-1}\restriction_{e_i}$ if the endpoints $u$ and $v$ of $e_i$ have the same root, i.e. $r(v) = r(u)$.
       \item If they have the same root, skip this edge and set $\mathcal{L}_i \restriction_{E^- \cup X} = \mathcal{L}_{i-1} \restriction_{E^- \cup X}$.
       \item If they have different roots, compute $\mathcal{L}_i \restriction_{E^-\cup X} = \mathsf{join}(\mathcal{L}_{i-1} \restriction_{E^-\cup X}, e_i)$.
       \end{enumerate}
    \item Check from labels how many connected components $\mathcal{L}_k \restriction_{E^-\cup X}$ has. If the number of roots is one, output that $X$ is independent, otherwise output that $X$ is not independent.
   \end{enumerate}
Note that since $T_1^*$ is connected, the final edge set $B \setminus X$ can only have multiple connected components due to removal of edges in $X$. Thus, the node $v$ will locally see all connected components of $B \setminus X$ from $\mathcal{L}_k \restriction_{E^- \cup X}$.
\end{proof}

\paragraph{Algorithm for handling weight decrements.}
The algorithm for batch dynamic minimum spanning tree under weight decrements is as follows:
\begin{enumerate}
\item Each node $v$ computes $\mathcal{L} \restriction_{E(v)}$ from the auxiliary state.
\item Broadcast $M(e)$ for each $e \in E^-$ using the broadcast tree $\mathcal{B}$.
\item Use the maximum matroid basis algorithm over $\mathcal{N}$ to compute $B^*$.
\item Each node $v$ locally computes $\mathcal{L}_1 \restriction_{E(v)}$ by applying the $\mathsf{cut}$ operation on each edge in $B^* \cap T_1^*$ in lexicographical order, starting from $\mathcal{L} \restriction_{E(v)}$.
\item Each node $v$ locally computes $\mathcal{L}_2 \restriction_{E(v)}$ by applying the $\mathsf{join}$ operation on each edge in $E^- \cap B^*$ in lexicographical order, starting from $\mathcal{L}_1 \restriction_{E(v)}$.
\item Each node $v$ outputs local auxiliary state $x_2(v)$ corresponding to $\mathcal{L}_2$.
\end{enumerate}

\begin{lemma}
  There is an algorithm that solves batch dynamic minimum-weight spanning trees under edge weight decrements in $O(\alpha + D)$ rounds.
  \label{lemma:weight-decr-alg}
\end{lemma}

\begin{proof}
By Lemma~\ref{lemma:mst-auxiliary}, Step (1)~of the algorithm can be done in $O(1)$ rounds, and by Lemma~\ref{lemma:primitives}, Step~(2) can be done in $O(\alpha + D)$ rounds. Step~(3) can be implemented in $O(\alpha + D)$ rounds by Theorem~\ref{thm:matroid-basis} and Lemma~\ref{lemma:mst-dual-matroid}, and after Step~(3) all nodes have learned the set $B^*$. Since all nodes apply the same operations to the Eulerian tour forest in the same order in Steps~(4) and (5), all nodes will produce compatible auxiliary states in Step~(6).
\end{proof}


\section{Lower bounds}\label{sec:lower}
In this section, we investigate lower bounds for the batch dynamic \congest model. 
We start with some necessary preliminaries in \sectionref{subsec:lbprelim} on two-party communication complexity~\cite{KushilevitzN97}, followed by our lower bound framework in \sectionref{subsec:lbf}, which we instantiate in \sectionref{subsec:lbinst}.
Finally, we give a lower bound for the minimum spanning tree problem in \sectionref{subsec:lbmst} by adapting arguments from \citet{dassarma12}.

\subsection{Two-party communication complexity}\label{subsec:lbprelim}

Let $f \colon \{0,1\}^{k}\times\{0,1\}^{k} \to \{0,1\}$ be a Boolean function. In~the two-party communication game on~$f$, there are two players who receive a private $k$-bit strings $x_0$ and $x_1$ as inputs, and their task is to have at least one of the players compute $f(x_0, x_1)$. 
The players follow a predefined protocol, and the complexity of a protocol is the maximum, over all $k$-bit inputs, of number of bits the parties exchange when executing the protocol on the input.
The \emph{deterministic communication complexity} $\cc(f)$ of a function $f$ is the minimal complexity of a protocol for computing $f$. Similarly, the \emph{randomised communication complexity} $\rcc(f)$ is the worst-case complexity of protocols, which compute $f$ with probability at least $2/3$ on all inputs, even if the players have access to a source of shared randomness.

While our framework is generic, all the reductions we use are based on \emph{set disjointness} lower bounds. In set disjointness over universe of size $k$, denoted by $\mathsf{DISJ}_k$, both players inputs are $x_0, x_1 \in \{ 0, 1 \}^k$, and the task is to decide whether the inputs are disjoint, i.e. $\mathsf{DISJ}_k(x_0, x_1) = 1$ if for all $i \in \{ 1, 2, \dotsc, k \}$ either $x_0(i)=0$ or $x_1(i) = 0$, and $\mathsf{DISJ}_k(x_0, x_1) = 0$ otherwise. It is known~\cite{KushilevitzN97,Razborov92} that 
\[ \cc(\mathsf{DISJ}_k) = \Omega(k)\, \hspace{5mm} \text{and} \hspace{5mm} \rcc(\mathsf{DISJ}_k) = \Omega(k)\,.\]

\subsection{Lower bound framework}\label{subsec:lbf}

For proving lower bounds for batch dynamic algorithms, we use the standard \congest lower bound framework of \emph{lower bound families} (e.g. \cite{drucker2013power,abboud2016near}). This allows us to translate existing \congest lower bound constructions to batch dynamic \congest
; however, we need a slightly different definition of lower bound families to account for our setting.

\newcommand{\lbfparam}{\alpha}

\begin{definition}\label{def:lower-bound-family}
For $\alpha\in\mathbb{N}$, let $f_\lbfparam \colon \{ 0, 1 \}^{2k(\lbfparam)} \to \{ 0, 1 \}$ and $s, C \colon \mathbb{N} \to \mathbb{N}$ be functions and $\Pi$ a predicate on labelled graphs. Suppose that there exists a constant $\lbfparam_0$ such that for all $\lbfparam > \lbfparam_0$ and $x_0,x_1 \in \{0,1\}^{k(\lbfparam)}$ there exists a labelled graph
$(G(\lbfparam),\ell(\alpha,x_0,x_1))$ satisfying the following properties:
\begin{enumerate}[noitemsep]
    \item $(G(\lbfparam),\ell(\alpha,x_0,x_1))$ satisfies $\Pi$ iff $f(x_0,x_1)=1$, 
    \item $G(\alpha) = (V_0 \cup V_1, E)$, where 
        \begin{enumerate}[label=(\alph*),noitemsep]
            \item $V_0$ and $V_1$ are disjoint and $\size{V_0 \cup V_1} = s(\lbfparam)$,
            \item the cut between $V_0$ and $V_1$ has size at most $C(\lbfparam)$,
        \end{enumerate}
    \item $\ell(\alpha,x_0,x_1) \colon E \to \Sigma$ is an edge labelling such that
        \begin{enumerate}[label=(\alph*),noitemsep]
            \item there are at most $\lbfparam$ edges whose labels depend on $x_0$ and $x_1$,
            \item for $i \in \{ 0, 1 \}$, all edges whose label depend on $x_i$ are in $E \cap V_i \times V_i$, and
            \item labels on all other edges do not depend on $x_0$ and $x_1$.
        \end{enumerate}
\end{enumerate}
We then say that $\mathcal{F} = (\mathcal{G}(\lbfparam))_{\lbfparam > \lbfparam_0}$ is a \emph{family of lower bound graphs for $\Pi$}, where 
\[\mathcal{G}(\lbfparam) = \bigl\{ (G(\lbfparam),\ell(\alpha,x_0,x_1)) \colon x_0, x_1 \in \{0,1\}^{k(\lbfparam)} \bigr\}\,.\]
\end{definition}

\paragraph{Extensions.} Since our aim is to prove lower bounds that depend on number of input changes $\alpha$ independently of the number of nodes $n$, we need to construct lower bounds where $\alpha$ can be arbitrarily small compared to $n$. We achieve this by embedding the lower bound graphs into a larger graph; this requires that the problem we consider has the following property.

\begin{definition}
Let $\Pi$ be a problem on labelled graphs. We say that $\Pi$ has the \emph{extension property with label $\gamma$} if $\gamma \in \Gamma$ is an input label such that for any labelled graph $(G, \ell)$, attaching new nodes and edges with label $\gamma$ does not change the output of the original nodes.
\end{definition}

\paragraph{Lower bound theorems.} We now present our lower bound framework, which we will instantiate in the next Section~\ref{subsec:lbinst}. We first show the following general version of the lower bound result.

\begin{theorem}\label{thm:lower-bound-main}
Let $\Pi$ be a problem, assume there is a family of lower bound graphs $\F$ for $\Pi$ and that $\Pi$ has the extension property, and let $L \colon \mathbb{N} \to \mathbb{N}$ be a function satisfying $L(\alpha) \ge s(\alpha)$. Let $\A$ be a deterministic batch dynamic algorithm that solves $\Pi$ in $T(\alpha,n)$ rounds for all $\alpha$ satisfying $n \ge L(\alpha)$ on batch dynamic \congest with bandwidth $b(n)$. Then we have
    \[
    T(\alpha, L(\alpha)) = \Omega\left( \frac{\cc(f_\alpha)}{C(\alpha) b\bigl(L(\alpha)\bigr)} \right)\,.
    \]
If $\A$ is a Monte Carlo algorithm with running time $T(\alpha,n)$ rounds and success probability at least $2/3$, or a Las Vegas algorithm with running time $T(\alpha,n)$ in either expectation or with probability at least $2/3$, then we instead have
    \[
    T(\alpha, L(\alpha)) = \Omega\left( \frac{\rcc(f_\alpha)}{C(\alpha) b\bigl(L(\alpha)\bigr)} \right)\,.
    \]

\end{theorem}

\begin{proof}
First consider the case of deterministic $\A$. We convert $\A$ into a two-player protocol computing $f_\alpha(x_0,x_1)$. Given inputs $x_0, x_1 \in \{ 0, 1 \}^{k(\alpha)}$, the players perform the following steps:
\begin{enumerate}[noitemsep]
    \item Both players construct the graph $G(\alpha)$ and a labelling $\ell$ such that $\ell$ agrees with $\ell(\alpha,x_0,x_1)$ on all labels that do not depend on $x_0$ and $x_1$, and other labels are set to some default label agreed to beforehand.
    \item Add new nodes connected to an arbitrary node with edges labelled with the extension label $\gamma$ to $(G(\alpha), \ell)$ to obtain $(G^*,\ell^*)$ where $G^*$ has $n = L(\alpha)$ nodes; since we assume $L(\alpha) \ge s(\alpha)$, this is possible.
    \item Simulate $\A$ on $G^*$, with player 0 simulating nodes in $V_0$ and player 1 simulating nodes in $V_1$:
      \begin{enumerate}[label=(\alph*),noitemsep]
        \item Both players construct a global auxiliary state $x \in \xi(G^*,\ell^*)$; since both players know $(G^*,\ell^*)$, they can do this locally.
        \item Player $i$ constructs a new partial labelling by changing the labels on their subgraph to match $\ell(\alpha, x_0, x_1)$. This defines a global labelling $\ell_1^*$, which differs from $\ell^*$ by on at most $\alpha$ edges. Players now simulate $\A(G^*, \ell^*, \ell^*_1, x)$ to obtain a new auxiliary state $x_1$; players locally simulate their owned nodes and messages between them, and send the messages that would cross the cut between $V_0$ and $V_1$ to each other.
    \end{enumerate}
    \item Players infer from $x_1$ whether $\Pi$ is satisfied, and produce the output $f_\alpha(x_0,x_1)$ accordingly.
\end{enumerate}
Each round, the algorithm $\A$ sends at most $2 b(n) = 2 b\bigl(L(\alpha)\bigr)$ bits over each edge, so the total number of bits players need to send to each other during the simulation is at most $2 b\bigl(L(\alpha)\bigr) C(\alpha) T\bigl(\alpha, L(\alpha)\bigr)$. Since the above protocol computes $f_\alpha$, we have for $\alpha > \alpha_0$ that $2 b\bigl(L(\alpha)\bigr) C(\alpha) T\bigl(\alpha, L(\alpha)\bigr) \ge \cc(f_\alpha)$, which implies
\[ T\bigl(\alpha, L(\alpha)\bigr) \ge \frac{\cc(f_\alpha)}{2 C(\alpha) b\bigl(L(\alpha)\bigr)}\,.\]

For randomised algorithms, we can directly apply same argument. If $\A$ is a Monte Carlo algorithm with success probability at least $2/3$, then the simulation gives correct result with probability at least $2/3$. If $\A$ is a Las Vegas algorithm that terminates in $T(\alpha,n)$ rounds with probability at least $2/3$, we can simulate $\A$ for $T(\alpha, n)$ rounds and give a random output if it does not terminate by that point; this succeeds in solving set disjointness with probability at least $2/3$. Likewise, if $\A$ has expected running time $T(\alpha, n)$, it suffices to simulate it for $3 T(\alpha, n)$ rounds. In all cases, we get
\[ T\bigl(\alpha, L(\alpha)\bigr) \ge \frac{\rcc(f_\alpha)}{2 C(\alpha) b\bigl(L(\alpha)\bigr)}\]
as desired.
\end{proof}

In practice, we use the following, simpler version of Theorem~\ref{thm:lower-bound-main} for our lower bounds.
Specifically, we assume the standard $\Theta(\log n)$ bandwidth and no dependence on $n$ in the running time; however, one can easily see that allowing e.g.~$\poly \log n$ factor in the running time will only weaker the lower bound by $\poly \log \alpha$ factor.

\begin{corollary}\label{thm:lower-bound-easy}
Let $\Pi$ be a problem, assume there is a family of lower bound graphs $\F$ for $\Pi$ and that $\Pi$ has the extension property, and let $\varepsilon > 0$ be a constant such that $s(\alpha) \le \alpha^{1/\varepsilon}$. Let $\A$ be a deterministic batch dynamic algorithm that solves $\Pi$ in $T(\alpha)$ rounds independent of $n$ for all $\alpha \le n^{\varepsilon}$ on batch dynamic \congest.
Then we have
    \[
    T(\alpha)  = \Omega\left( \frac{\cc(f_\alpha)}{C(\alpha) \log \alpha} \right)\,.
    \]
If $\A$ is a Monte Carlo algorithm with running time $T(\alpha)$ rounds and success probability at least $2/3$, or a Las Vegas algorithm with running time $T(\alpha)$ in either expectation or with probability at least $2/3$, then we instead have
    \[
    T(\alpha)  = \Omega\left( \frac{\rcc(f_\alpha)}{C(\alpha) \log \alpha} \right)\,.
    \]
\end{corollary}

Note the role of $\varepsilon$ and $s$ in the claim; the lower bounds in terms of $\alpha$ only work in a regime where $\alpha$ is sufficiently small compared to $n$. The limit where the lower bound stops working usually corresponds to the complexity of computing the solution from scratch, that is, if $\alpha$ is sufficiently large, then recomputing everything is cheap in terms of the parameter $\alpha$. On the other hand, we can make $\varepsilon$ arbitrarily small, so the lower bound holds even under a promise of small batch size, e.g, $\alpha \le n^{1/1000}$.

\subsection{Instantiations}\label{subsec:lbinst}

We now obtain concrete lower bounds by plugging in prior constructions for lower bound families into our framework. These constructions, originally used for \congest lower bounds, are parameterised by the number of nodes $n$, but transforming them to the form used in Definition~\ref{def:lower-bound-family} is a straightforward reparameterisation.

\paragraph{Clique detection.}

In $k$-clique detection for fixed $k$, the input labelling $\ell \colon V \to \{ 0, 1 \}$ defines a subgraph $H$ of $G$, and each node has to output $1$ if they are part of a $k$-clique in $H$, and $0$ otherwise. The corresponding graph property is $k$-clique freeness, and $k$-clique detection has the extension property with label~$0$.
\begin{itemize}
    \item \emph{Lower bound family.} For fixed $k \ge 4$, \citet{czumaj2019detecting} give a family of lower bound graphs with parameters
    \[ f_\alpha = \mathsf{DISJ}_{\Theta(\alpha)}, \hspace{2mm} s(\alpha) = \Theta(\alpha^{1/3}), \hspace{2mm}  C(\alpha) = \Theta(\alpha^{3/4})\,. \]
    The lower bound given by Corollary~\ref{thm:lower-bound-easy} is $\Omega(\alpha^{1/4}/\log \alpha)$ for any $\alpha$.
\end{itemize}

\paragraph{Cycle detection.}
Next we consider $k$-cycle detection for fixed $k$: the input labelling $\ell \colon V \to \{ 0, 1 \}$ defines a subgraph $H$ of $G$, and each node has to output $1$ if they are part of a $k$-cycle in $H$, and $0$ otherwise. The corresponding graph property is $k$-cycle freeness, and $k$-cycle detection clearly has the extension property with label $0$.
For different parameters $k$, we obtain the lower bounds from prior constructions as follows.
\begin{itemize}
    \item For $4$-cycle detection, \citet{drucker2013power} give a family of lower bound graphs with parameters
    \[ f_\alpha = \mathsf{DISJ}_{\Theta(\alpha)}, \hspace{2mm} s(\alpha) = \Theta(\alpha^{2/3}), \hspace{2mm}  C(\alpha) = \Theta(\alpha^{2/3})\,. \]
    The lower bound given by Corollary~\ref{thm:lower-bound-easy} is $\Omega(\alpha^{1/3}/\log \alpha)$ for $\alpha = O(n^{3/2})$.
    \item For $(2k+1)$-cycle detection for $k \ge 2$, Drucker et~al.~\cite{drucker2013power} give a family of lower bound graphs with 
    \[ f_\alpha = \mathsf{DISJ}_{\Theta(\alpha)}, \hspace{2mm} s(\alpha) = \Theta(\alpha^{1/2}), \hspace{2mm}  C(\alpha) = \Theta(\alpha^{1/2})\,. \]
    The lower bound given by Corollary~\ref{thm:lower-bound-easy} is $\Omega(\alpha^{1/2}/\log \alpha)$ for $\alpha = O(n^{2})$.
    \item For $2k$-cycle detection for $k \ge 3$, \citet{korhonen2017deterministic} give a family of lower bound graphs~with
    \[ f_\alpha = \mathsf{DISJ}_{\Theta(\alpha)}, \hspace{2mm} s(\alpha) = \Theta(\alpha), \hspace{2mm}  C(\alpha) = \Theta(\alpha^{1/2})\,. \]
    The lower bound given by Corollary~\ref{thm:lower-bound-easy} is $\Omega(\alpha^{1/2}/\log \alpha)$ for $\alpha = O(n)$.
\end{itemize}

\paragraph{Diameter and all-pairs shortest paths.}
In diameter computation, the input labelling $\ell \colon V \to \{ 0, 1 \}$ defines a subgraph $H$ of $G$, and each node has to output the diameter of their connected component in $H$. Again, diameter computation has the extension property with label $0$.
For exact and approximate diameter computation, we use the sparse lower bound constructions of \citet{abboud2016near}:
\begin{itemize}
    \item For distinguishing between graphs of diameter $4$ and $5$, there is a family of lower bound graphs with parameters
    \[ f_\alpha = \mathsf{DISJ}_{\Theta(\alpha)}, \hspace{4mm} s(\alpha) = \Theta(\alpha), \hspace{4mm}  C(\alpha) = \Theta(\log \alpha)\,. \]
    The lower bound given by Corollary~\ref{thm:lower-bound-easy} is $\Omega(\alpha/\log^2 \alpha)$ for $\alpha = O(n)$. This implies a lower bound for exact diameter computation.
    \item For distinguishing between graphs of diameter $4k+2$ and $6k+1$, there is a family of lower bound graphs with parameters
    \[ f_\alpha = \mathsf{DISJ}_{\Theta(\alpha)}, \hspace{1.5mm} s(\alpha) = O(\alpha^{1+\delta}), \hspace{1.5mm}  C(\alpha) = \Theta(\log \alpha), \]
    for any constant $\delta > 0$.
    The lower bound given by Corollary~\ref{thm:lower-bound-easy} is $\Omega(\alpha/\log^2 \alpha)$ for $\alpha = O(n^{1/(1+\delta)})$ for any constant $\delta >0$. This implies a lower bound for $(3/2 - \varepsilon)$-approximation of diameter for any constant $\varepsilon > 0$.
\end{itemize}
A trivial $\Omega(D)$ lower bound holds even for $(3/2 - \varepsilon)$-approximation in the worst case (e.g.\ a cycle).

In all-pairs shortest paths problem, the input labelling gives a weight $w(e) \in \{ 0, 1, 2, \dotsc, n^C \} \cup \{ \infty \}$ for each edge $e \in E$, and each node node $v$ has to output the distance $d(v,u)$ for each other node $u \in V \setminus \{ v \}$.  Exact or $(3/2 - \varepsilon)$-approximate solution to all-pairs shortest paths can be used to recover exact or $(3/2 - \varepsilon)$-approximate solution to diameter computation, respectively, in $O(D)$ rounds, so the lower bounds also apply to batch dynamic all-pairs shortest paths.

\subsection{Lower bound for minimum spanning tree}\label{subsec:lbmst}

The \congest lower bound for minimum spanning tree does not fall under the family of lower bound graphs construction used above; indeed, one can show that it is in fact impossible to prove \congest lower bounds for minimum spanning tree using a \emph{fixed-cut simulation} (see \citet{bacrach2019hardness}). However, we can adapt the more involved simulation argument of \citet{dassarma12} to obtain a near-linear lower bound for batch dynamic~MST; note that $\Omega(D)$ lower bound holds trivially for the problem.

Again, we first prove a general version of the lower bound theorem first.

\begin{theorem}\label{thm:mst-lower-bound}
Let $L \colon \mathbb{N} \to \mathbb{N}$ be a function satisfying $L(\alpha) \ge \alpha^{2}$. Let $\A$ be a deterministic batch dynamic algorithm or a randomised batch dynamic algorithm as in Theorem~\ref{thm:lower-bound-main} that solves MST in $T(\alpha,n) + O(D)$ rounds for all $\alpha$ satisfying $n \ge L(\alpha)$ on batch dynamic \congest with bandwidth $b(n)$. Then we have
    \[
    T(\alpha, L(\alpha)) = \Omega\biggl( \frac{\alpha}{b\bigl(L(\alpha)\bigr) \log \alpha} \biggr)\,.
    \]
\end{theorem}

\begin{proof}
We follow the proof of \citet{dassarma12} with the same modifications to standard \congest lower bounds as in Theorem~\ref{thm:lower-bound-main}. We construct labelled graphs $(G_\alpha, \ell_\alpha)$ as follows, with $\ell_\alpha$ encoding the edge weights of the graph: 
\begin{itemize}[noitemsep]
    \item We start with two terminal nodes $a$ and $b$.
    \item We add $\alpha/2$ paths $P_1, P_2, \dotsc, P_{\alpha/2}$ of length $\alpha$, with all edges having weight $0$. Each path $P_i$ consists of nodes $p_{i,1}, p_{i,2}, \dotsc, p_{i,\alpha}$, and we refer to the set $\{ p_{1,j}, p_{2,j}, \dotsc, p_{\alpha/2, j} \}$ as \emph{column} $j$.
    \item We connect $p_{i, 1}$ to $a$ and and $p_{i,\alpha}$ to $b$ for all $i$. These edges have weight $0$.
    \item We add a balanced binary tree with $\alpha$ leaves, with all edges weight $0$. We connect the first leaf to $a$ with weight-$0$ edge, and the last leaf to $b$ with weight-$0$ edge.
    \item We connect $i$th leaf of the tree to $i$th edge on each path $P_j$ with weight-$1$ edge.
    \item Finally, we add new nodes connected by weight-$0$ edges to $a$ to satisfy $n \ge L(\alpha)$; since we assume $L(\alpha) \ge \alpha^2$, this is always possible.
\end{itemize}
See Figure~\ref{fig:mst-lower-bound} for an example.

\begin{figure}
\begin{center}
  \includegraphics[page=2, scale=0.9]{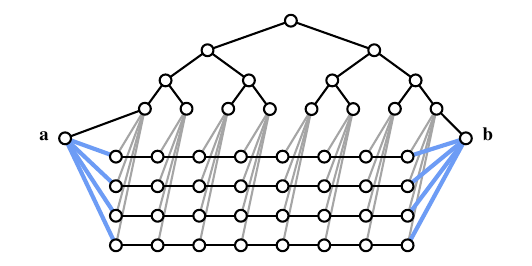}
\end{center}
  \caption{Instance of the graph $G_\alpha$ for $\alpha = 8$ used in the minimum spanning tree lower bound. Black edges have weight $0$, grey edges have weight $1$ and blue edges are used to encode the set disjointness instance.\label{fig:mst-lower-bound}}
\end{figure}

We now turn the algorithm $\A$ into a two-player protocol for solving $\mathsf{DISJ}_{\alpha/2}$. Given inputs $x_0, x_1 \in \{ 0, 1 \}^{\alpha/2}$, the players first construct $(G_\alpha, \ell_\alpha)$, and construct a global auxiliary state $x \in \xi(G_\alpha,\ell_\alpha)$; since both players know the $(G_\alpha,\ell_\alpha)$, they can do this locally. The players then locally change the labels according to the inputs $x_0$ and $x_1$:
\begin{itemize}[noitemsep]
    \item player $0$ sets the weight on the edge from $a$ to $p_{i,1}$ to weight $x_0(i)$ for $i = 1, 2, \dotsc, \alpha/2$, and
    \item player $1$ sets the weight of the edge from $b$ to $p_{i,\alpha}$ to $x_1(i)$  for $i = 1, 2, \dotsc, \alpha/2$.
\end{itemize}
This defines a new global labelling $\ell^*$. The players now simulate the execution $\A(G_\alpha, \ell_\alpha, \ell^*, x)$ in a distributed manner; note in particular that players do not know the whole labelling $\ell^*$.

We assume that $T \le \alpha/2$, as otherwise we already $T > \alpha/2$ and we are happy. The simulation proceeds in steps $t = 1, 2, \dotsc, T$, where $T$ is the running time of $\A$ on the instance.
\begin{enumerate}[noitemsep]
    \item In step $t$ of the iteration, player $0$ simulates node $a$, columns $1$ to $\alpha - t$, and the smallest subtree of the binary tree that includes children from $1$ to $\alpha - t$. Dually, player $1$ simulates node $b$, columns $i+t$ to $\alpha$, and the smallest subtree of the binary tree that includes children from $t+1$ to $\alpha$.
    \item At the start of the simulation, both players know the local inputs of all the nodes they are simulating, since they are not simulating the nodes whose incident labels were changed by the other player.
    \item At step $t+1$, players simulate one round of $\A$. We describe how player $0$ does the simulation; player $1$ acts in symmetrical way.
    \begin{enumerate}[label=(\alph*)]
        \item Since the set of nodes player $0$ simulates in round $t+1$ is a subset of nodes simulated in step $t$, player $0$ knows the full state of all the nodes it is simulating.
        \item For path nodes simulated by player $0$, their neighbours were simulated in the previous round by player $0$, so their incoming messages can be determined locally.
        \item For binary tree nodes, there can be neighbours that were not simulated in the previous round by player $0$. However, since $T \le \alpha/2$, these are simulated by player $1$, and player $1$ sends their outgoing messages to player $0$. Since the height of the binary tree is $O(\log \alpha)$ and player $0$ simulates a subtree of the binary tree, there are $O(\log \alpha)$ nodes that need to receive their neighbours' messages from player $1$. Thus player $1$ has to send $O(b(L(\alpha)) \log \alpha)$ bits to player $0$ to complete one iteration of the simulation.
    \end{enumerate}
\end{enumerate}
In total, the simulation of the execution of $\A(G_\alpha, \ell_\alpha, \ell^*, x)$ uses at most $C T b(L(\alpha)) \log \alpha$ bits of communication for constant $C$. One can verify that the minimum spanning tree in $(G_\alpha, \ell^*)$ has weight $0$ if $x_0$ and $x_1$ are disjoint, and weight at least $1$ if they are not disjoint, so the players can determine the disjointness from the output of $\A$. For deterministic $\A$, this implies that $C T b(L(\alpha)) \log \alpha  \ge \cc(\mathsf{DISJ}_{\alpha/2})$, and thus
\[ T \ge \frac{\cc(\mathsf{DISJ}_{\alpha/2})}{C b\bigl(L(\alpha)\bigr)\log \alpha} = \frac{C' \alpha }{b\bigl(L(\alpha)\bigr) \log \alpha}\]
for a constant $C'$. For randomised $\A$, we similarly get
\[ T \ge \frac{\rcc(\mathsf{DISJ}_{\alpha/2})}{C b\bigl(L(\alpha)\bigr)\log \alpha} = \frac{C' \alpha }{b\bigl(L(\alpha)\bigr) \log \alpha}\]
by the same argument as in the proof of Theorem~\ref{thm:lower-bound-main}.
Finally, since the diameter of $G_\alpha$ is $O(\log n)$, we have that for sufficiently large $\alpha$, we have $T(\alpha, L(\alpha)) \ge T/2$, and the claim follows.
\end{proof}

The general theorem implies the following simplified claim:

\begin{corollary}\label{thm:mst-lower-bound-easy}
Let $\A$ be a deterministic batch dynamic algorithm or a randomised batch dynamic algorithm as in Corollary~\ref{thm:lower-bound-easy} that solves MST in $T(\alpha) + D$ rounds independent of $n$ for all $\alpha \le n^{\varepsilon}$ on batch dynamic \congest with bandwidth $\Theta(\log n)$, where $\varepsilon \le 1/2$ is a constant. Then we have $T(\alpha) = \Omega(\alpha/\log^2 \alpha )$.
\end{corollary}


\section{Batch dynamic congested clique}\label{sec:dcm}

If we set the communication graph $G = (V,E)$ to be a clique, we obtain a batch dynamic version of the \emph{congested clique}~\cite{lotker2005mst} as a special case of our batch dynamic \congest model. This is in many ways similar to the batch dynamic versions of the $k$-machine and MPC models~\cite{italiano2019dynamic,DhulipalaDKPSS20,NowickiO20,GilbertL20}; however, whereas the these usually consider setting where the number of nodes~$k$ is much smaller than $n$, the setting with $k = n$ is qualitatively different. For example, a minimum spanning tree can be computed from scratch in $O(1)$ rounds in the congested clique~\cite{nowicki2019deterministic}, so recomputing from scratch is optimal also for input-dynamic algorithms.

In this section, we briefly discuss the batch dynamic congested clique, and in particular highlight \emph{triangle counting} (and hence \emph{triangle detection}) as an example of problem admitting a non-trivial batch dynamic algorithm in this setting.

\subsection{Universal upper bound}
First, we
make the simple observation
that the fully-connected communication topology gives faster universal upper bound than \theoremref{thm:universal-alg}.

\begin{theorem}\label{thm:universal-alg-clique}
  For any problem $\Pi$, there is a batch dynamic congested clique algorithm that runs in $O(\lceil \alpha / n \rceil)$ rounds and uses $O(m \log n)$ bits of auxiliary state.
\end{theorem}

\begin{proof}
Use the same algorithm as in \theoremref{thm:universal-alg}; the claim follows by observing that the message set $M$ can be learned by all nodes in $O(\lceil \alpha / n \rceil)$ rounds using standard congested clique routing techniques~\cite{lenzen2013optimal}.
\end{proof}

\subsection{Batch dynamic matrix multiplication and triangle detection}

As an example of a problem that has non-trivial batch dynamic algorithms in congested clique, we consider the following \emph{dynamic matrix multiplication} task. As input, we are given two $n \times n$ matrices $S$, $T$ so that each node $v$ receives row $v$ of $S$ and column $v$ of $T$, and the task is to compute the product matrix $P = ST$ so that node $v$ outputs row $v$ of $P$. Concretely, we assume that the input label on edge $\{ u, v \}$ the matrix entries $S[v,u]$, $S[u,v]$, $T[v,u]$ and $T[u,v]$. Note that in the dynamic version of the problem, the parameter $\alpha$ is an upper bound for changes to both matrices.

For matrix $S$, let \emph{density $\rho_S$} of $S$ be the smallest integer $\rho$ such that the number of non-zero elements in $S$ is less than $\rho n$. We use the following result:

\begin{theorem}[\cite{censor-hillel2019fast,censor-hillel2020sparse}]\label{thm:clique-mm}
There is a congested clique algorithm that computes the product $P = ST$ in
$O\bigl( (\rho_S \rho_T)^{1/3}/n^{1/3} + 1 \bigr)$
rounds.
\end{theorem}

We use \theoremref{thm:clique-mm} to obtain a non-trivial dynamic batch algorithm for matrix multiplication. This in turn implies an upper bound for triangle counting by a standard reduction.

\begin{theorem}
There is a batch dynamic algorithm for matrix multiplication in congested clique that runs in $O\bigl((\alpha / n)^{1/3} + 1 \bigr)$ rounds and uses $O(n \log n)$ bits of auxiliary state. \label{thm:clique-dynamic-mm}
\end{theorem}

\begin{proof}
Consider input matrices $S_1$ and $T_1$ and updated input matrices $S_2$ and $T_2$. As auxiliary data $x(v)$, each node $v$ keeps the row $v$ of the matrix $P_1 = S_1 T_1$.

We can write
\[ S_2 = S_1 + \Delta_S\,, \hspace{20mm} T_2 = T_1 + \Delta_T\,, \]
where $\Delta_S$ and $\Delta_T$ are matrices with at most $\alpha$ non-zero elements, which implies their density is at most $\lceil \alpha/n \rceil$. Thus, we can write the product $P_2 = S_2 T_2$ as
\begin{align*}
P_2 & = (S_1 + \Delta_S) (T_1 + \Delta_T)\\
    & = S_1 T_1 + \Delta_S T_1 + S_1 \Delta_T + \Delta_S \Delta_T\\
    & = P_1 + \Delta_S T_1 + S_1 \Delta_T + \Delta_S \Delta_T\,.
\end{align*}
That is, it suffices to compute the products $\Delta_S T_1$,  $S_1 \Delta_T$ and $\Delta_S \Delta_T$ to obtain $P_2$; by Theorem~\ref{thm:clique-mm}, this can be done in $O\bigl((\alpha / n)^{1/3} + 1 \bigr)$ rounds.
\end{proof}

\begin{corollary}
There is a batch dynamic algorithm for triangle counting in congested clique that runs in $O\bigl((\alpha / n)^{1/3} + 1 \bigr)$ rounds and uses $O(n \log n)$ bits of auxiliary state.
\end{corollary}

\subsection*{Acknowledgements}

We thank Jukka Suomela for discussions.
We also thank our shepherd Mohammad Hajiesmaili and the reviewers for their time and suggestions on how to improve the paper.
This project has received funding from the European Research Council (ERC) under the European Union's Horizon 2020 research and innovation programme (grant agreement No 805223 ScaleML), from the European Union's Horizon 2020 research and innovation programme under the Marie Sk\l{}odowska--Curie grant agreement No. 840605, from the Vienna Science and Technology Fund (WWTF) project WHATIF, ICT19-045, 2020-2024, and from the Austrian Science Fund (FWF) and netIDEE SCIENCE project P 33775-N.

\DeclareUrlCommand{\Doi}{\urlstyle{same}}
\renewcommand{\UrlFont}{\footnotesize\sf}
\renewcommand{\doi}[1]{\href{http://dx.doi.org/#1}{\footnotesize\sf doi:\Doi{#1}}}

\bibliographystyle{plainnat}
\bibliography{dynamic-congest}

\begin{thebibliography}{77}
\providecommand{\natexlab}[1]{#1}
\providecommand{\url}[1]{\texttt{#1}}
\expandafter\ifx\csname urlstyle\endcsname\relax
  \providecommand{\doi}[1]{doi: #1}\else
  \providecommand{\doi}{doi: \begingroup \urlstyle{rm}\Url}\fi

\bibitem[Abboud and Williams(2014)]{AbboudW14}
Amir Abboud and Virginia~Vassilevska Williams.
\newblock Popular conjectures imply strong lower bounds for dynamic problems.
\newblock In \emph{55th {IEEE} Annual Symposium on Foundations of Computer
  Science, {FOCS}}, pages 434--443, 2014.
\newblock \doi{10.1109/FOCS.2014.53}.

\bibitem[Abboud et~al.(2016)Abboud, Censor-Hillel, and Khoury]{abboud2016near}
Amir Abboud, Keren Censor-Hillel, and Seri Khoury.
\newblock Near-linear lower bounds for distributed distance computations, even
  in sparse networks.
\newblock In \emph{Proc.\ 30th International Symposium on Distributed Computing
  (DISC 2016)}, pages 29--42. Springer, 2016.
\newblock \doi{10.1007/978-3-662-53426-7\_3}.

\bibitem[Acar et~al.(2011)Acar, Cotter, Hudson, and
  T{\"u}rkoglu]{acar2011parallelism}
Umut~A Acar, Andrew Cotter, Benoit Hudson, and Duru T{\"u}rkoglu.
\newblock Parallelism in dynamic well-spaced point sets.
\newblock In \emph{Proc.\ 23rd annual ACM symposium on Parallelism in
  algorithms and architectures (SPAA 2011)}, pages 33--42, 2011.
\newblock \doi{10.1145/1989493.1989498}.

\bibitem[Acar et~al.(2017)Acar, Aksenov, and Westrick]{acar2017parallel}
Umut~A. Acar, Vitaly Aksenov, and Sam Westrick.
\newblock Brief announcement: Parallel dynamic tree contraction via
  self-adjusting computation.
\newblock In \emph{Proc.\ 29th ACM Symposium on Parallelism in Algorithms and
  Architectures (SPAA 2017)}, pages 275--277, New York, NY, USA, 2017.
  Association for Computing Machinery.
\newblock \doi{10.1145/3087556.3087595}.

\bibitem[Acar et~al.(2019)Acar, Anderson, Blelloch, and
  Dhulipala]{acar2019parallel}
Umut~A. Acar, Daniel Anderson, Guy~E. Blelloch, and Laxman Dhulipala.
\newblock Parallel batch-dynamic graph connectivity.
\newblock In \emph{Proc.\ 31st ACM Symposium on Parallelism in Algorithms and
  Architectures ({SPAA} 2019)}, pages 381--392, 2019.
\newblock \doi{10.1145/3323165.3323196}.

\bibitem[Ancona et~al.(2019)Ancona, Henzinger, Roditty, Williams, and
  Wein]{AnconaHRWW19}
Bertie Ancona, Monika Henzinger, Liam Roditty, Virginia~Vassilevska Williams,
  and Nicole Wein.
\newblock Algorithms and hardness for diameter in dynamic graphs.
\newblock In \emph{46th International Colloquium on Automata, Languages, and
  Programming, {ICALP}}, pages 13:1--13:14, 2019.
\newblock \doi{10.4230/LIPIcs.ICALP.2019.13}.

\bibitem[Assadi et~al.(2018)Assadi, Onak, Schieber, and
  Solomon]{assadi2018sublinear}
Sepehr Assadi, Krzysztof Onak, Baruch Schieber, and Shay Solomon.
\newblock Fully dynamic maximal independent set with sublinear update time.
\newblock In \emph{Proc.\ 50th Annual {ACM} {SIGACT} Symposium on Theory of
  Computing ({STOC})}, pages 815--826, 2018.
\newblock \doi{10.1145/3188745.3188922}.

\bibitem[Assadi et~al.(2019)Assadi, Onak, Schieber, and
  Solomon]{assadi2019sublinear-in-n}
Sepehr Assadi, Krzysztof Onak, Baruch Schieber, and Shay Solomon.
\newblock Fully dynamic maximal independent set with sublinear in $n$ update
  time.
\newblock In \emph{Proc.\ 30th Annual {ACM-SIAM} Symposium on Discrete
  Algorithms (SODA)}, pages 1919--1936, 2019.
\newblock \doi{10.1137/1.9781611975482.116}.

\bibitem[Avin et~al.(2020)Avin, Ghobadi, Griner, and
  Schmid]{sigmetrics20complexity}
Chen Avin, Manya Ghobadi, Chen Griner, and Stefan Schmid.
\newblock On the complexity of traffic traces and implications.
\newblock \emph{Proceedings of the {ACM} on Measurement and Analysis of
  Computing Systems}, 4\penalty0 (1):\penalty0 20:1--20:29, 2020.
\newblock \doi{10.1145/3379486}.

\bibitem[Awduche et~al.(2002)Awduche, Chiu, Elwalid, Widjaja, and
  Xiao]{awduche2002overview}
Daniel Awduche, Angela Chiu, Anwar Elwalid, Indra Widjaja, and XiPeng Xiao.
\newblock Overview and principles of internet traffic engineering.
\newblock Technical report, RFC 3272, 2002.

\bibitem[Awerbuch et~al.(1992)Awerbuch, Patt-Shamir, Peleg, and
  Saks]{awerbuch1992adapting}
Baruch Awerbuch, Boaz Patt-Shamir, David Peleg, and Michael Saks.
\newblock Adapting to asynchronous dynamic networks.
\newblock In \emph{Proc.\ 24th annual ACM symposium on Theory of computing},
  pages 557--570, 1992.
\newblock \doi{10.1145/129712.129767}.

\bibitem[Awerbuch et~al.(2008)Awerbuch, Cidon, and Kutten]{awerbuch2008optimal}
Baruch Awerbuch, Israel Cidon, and Shay Kutten.
\newblock Optimal maintenance of a spanning tree.
\newblock \emph{Journal of the ACM}, 55\penalty0 (4), September 2008.
\newblock \doi{10.1145/1391289.1391292}.

\bibitem[Bacrach et~al.(2019)Bacrach, Censor-Hillel, Dory, Efron, Leitersdorf,
  and Paz]{bacrach2019hardness}
Nir Bacrach, Keren Censor-Hillel, Michal Dory, Yuval Efron, Dean Leitersdorf,
  and Ami Paz.
\newblock Hardness of distributed optimization.
\newblock In \emph{Proc.\ of the 2019 ACM Symposium on Principles of
  Distributed Computing ({PODC} 2019)}, pages 238--247, 2019.
\newblock \doi{10.1145/3293611.3331597}.

\bibitem[Bamberger et~al.(2019)Bamberger, Kuhn, and Maus]{bamberger2019local}
Philipp Bamberger, Fabian Kuhn, and Yannic Maus.
\newblock Local distributed algorithms in highly dynamic networks.
\newblock In \emph{Proc.\ 33rd IEEE International Parallel and Distributed
  Processing Symposium (IPDPS 2019)}, 2019.
\newblock \doi{10.1109/IPDPS.2019.00015}.

\bibitem[Barenboim and Elkin(2010)]{barenboim2010sublogarithmic}
Leonid Barenboim and Michael Elkin.
\newblock Sublogarithmic distributed mis algorithm for sparse graphs using
  nash-williams decomposition.
\newblock \emph{Distributed Computing}, 22\penalty0 (5-6):\penalty0 363--379,
  2010.

\bibitem[Barenboim et~al.(2018)Barenboim, Elkin, and
  Goldenberg]{barenboim2018locally}
Leonid Barenboim, Michael Elkin, and Uri Goldenberg.
\newblock Locally-iterative distributed {$(\Delta+1)$}-coloring below
  {S}zegedy-{V}ishwanathan barrier, and applications to self-stabilization and
  to restricted-bandwidth models.
\newblock In \emph{Proc.\ {ACM} Symposium on Principles of Distributed
  Computing (PODC 2018)}, pages 437--446, 2018.
\newblock \doi{10.1145/3212734.3212769}.

\bibitem[Bernstein et~al.(2019)Bernstein, Forster, and
  Henzinger]{BernsteinFH19}
Aaron Bernstein, Sebastian Forster, and Monika Henzinger.
\newblock A deamortization approach for dynamic spanner and dynamic maximal
  matching.
\newblock In \emph{Proc.\ 30th Annual {ACM-SIAM} Symposium on Discrete
  Algorithms, {SODA}}, pages 1899--1918, 2019.
\newblock \doi{10.1137/1.9781611975482.115}.

\bibitem[Bhattacharya et~al.(2018)Bhattacharya, Henzinger, and
  Italiano]{BhattacharyaHI18}
Sayan Bhattacharya, Monika Henzinger, and Giuseppe~F. Italiano.
\newblock Dynamic algorithms via the primal-dual method.
\newblock \emph{Inf. Comput.}, 261\penalty0 (Part):\penalty0 219--239, 2018.
\newblock \doi{10.1016/j.ic.2018.02.005}.

\bibitem[Casteigts et~al.(2012)Casteigts, Flocchini, Quattrociocchi, and
  Santoro]{casteigts2012time}
Arnaud Casteigts, Paola Flocchini, Walter Quattrociocchi, and Nicola Santoro.
\newblock Time-varying graphs and dynamic networks.
\newblock \emph{{IJPEDS}}, 27\penalty0 (5):\penalty0 387--408, 2012.
\newblock \doi{10.1080/17445760.2012.668546}.

\bibitem[Censor{-}Hillel et~al.(2016)Censor{-}Hillel, Haramaty, and
  Karnin]{Censor-HillelHK16}
Keren Censor{-}Hillel, Elad Haramaty, and Zohar~S. Karnin.
\newblock Optimal dynamic distributed {MIS}.
\newblock In \emph{Proc.\ 2016 {ACM} Symposium on Principles of Distributed
  Computing ({PODC} 2016)}, pages 217--226, 2016.
\newblock \doi{10.1145/2933057.2933083}.

\bibitem[Censor-Hillel et~al.(2019)Censor-Hillel, Dory, Korhonen, and
  Leitersdorf]{censor-hillel2019fast}
Keren Censor-Hillel, Michal Dory, Janne~H. Korhonen, and Dean Leitersdorf.
\newblock Fast approximate shortest paths in the congested clique.
\newblock In \emph{Proc.\ 38nd ACM Symposium on Principles of Distributed
  Computing (PODC 2019)}, pages 74--83, 2019.
\newblock \doi{10.1145/3293611.3331633}.

\bibitem[Censor{-}Hillel et~al.(2020{\natexlab{a}})Censor{-}Hillel, Dafni,
  Kolobov, Paz, and Schwartzman]{CensorHillelDKPS20}
Keren Censor{-}Hillel, Neta Dafni, Victor~I. Kolobov, Ami Paz, and Gregory
  Schwartzman.
\newblock Fast deterministic algorithms for highly-dynamic networks.
\newblock \emph{CoRR}, abs/1901.04008, 2020{\natexlab{a}}.
\newblock URL \url{http://arxiv.org/abs/1901.04008}.

\bibitem[Censor{-}Hillel et~al.(2020{\natexlab{b}})Censor{-}Hillel,
  Leitersdorf, and Turner]{censor-hillel2020sparse}
Keren Censor{-}Hillel, Dean Leitersdorf, and Elia Turner.
\newblock Sparse matrix multiplication and triangle listing in the congested
  clique model.
\newblock \emph{Theoretical Computer Science}, 809:\penalty0 45--60,
  2020{\natexlab{b}}.
\newblock \doi{10.1016/j.tcs.2019.11.006}.

\bibitem[Cicerone et~al.(2003)Cicerone, Stefano, Frigioni, and
  Nanni]{cicerone2003fully}
Serafino Cicerone, Gabriele~Di Stefano, Daniele Frigioni, and Umberto Nanni.
\newblock A fully dynamic algorithm for distributed shortest paths.
\newblock \emph{Theoretical Computer Science}, 297\penalty0 (1):\penalty0
  83--102, 2003.
\newblock ISSN 0304-3975.
\newblock \doi{10.1016/S0304-3975(02)00619-9}.

\bibitem[Czumaj and Konrad(2019)]{czumaj2019detecting}
Artur Czumaj and Christian Konrad.
\newblock Detecting cliques in congest networks.
\newblock \emph{Distributed Computing}, 2019.
\newblock \doi{10.1007/s00446-019-00368-w}.

\bibitem[Das~Sarma et~al.(2012)Das~Sarma, Holzer, Kor, Korman, Nanongkai,
  Pandurangan, Peleg, and Wattenhofer]{dassarma12}
Atish Das~Sarma, Stephan Holzer, Liah Kor, Amos Korman, Danupon Nanongkai,
  Gopal Pandurangan, David Peleg, and Roger Wattenhofer.
\newblock Distributed verification and hardness of distributed approximation.
\newblock \emph{{SIAM} Journal on Computing}, 41:\penalty0 1235--1265, 2012.
\newblock \doi{10.1137/11085178X}.

\bibitem[Dhulipala et~al.(2020)Dhulipala, Durfee, Kulkarni, Peng, Sawlani, and
  Sun]{DhulipalaDKPSS20}
Laxman Dhulipala, David Durfee, Janardhan Kulkarni, Richard Peng, Saurabh
  Sawlani, and Xiaorui Sun.
\newblock Parallel batch-dynamic graphs: Algorithms and lower bounds.
\newblock In \emph{Proc.\ 2020 {ACM-SIAM} Symposium on Discrete Algorithms,
  {SODA}}, pages 1300--1319, 2020.
\newblock \doi{10.1137/1.9781611975994.79}.

\bibitem[Dijkstra(1974)]{dijkstra1974self-stabilization}
Edsger~W. Dijkstra.
\newblock Self-stabilizing systems in spite of distributed control.
\newblock \emph{Communications of the ACM}, 17\penalty0 (11):\penalty0
  643--644, 1974.
\newblock \doi{10.1145/361179.361202}.

\bibitem[Dolev(2000)]{dolev00self-stabilization}
Shlomi Dolev.
\newblock \emph{Self-Stabilization}.
\newblock Cambridge, MA, 2000.

\bibitem[Drucker et~al.(2014)Drucker, Kuhn, and Oshman]{drucker2013power}
Andrew Drucker, Fabian Kuhn, and Rotem Oshman.
\newblock On the power of the congested clique model.
\newblock In \emph{Proc.\ 33rd ACM Symposium on Principles of Distributed
  Computing (PODC 2014)}, pages 367--376, 2014.
\newblock \doi{10.1145/2611462.2611493}.

\bibitem[Du and Zhang(2018)]{DuZ2018}
Yuhao Du and Hengjie Zhang.
\newblock Improved algorithms for fully dynamic maximal independent set.
\newblock \emph{CoRR}, abs/1804.08908, 2018.
\newblock URL \url{http://arxiv.org/abs/1804.08908}.

\bibitem[Durfee et~al.(2019)Durfee, Gao, Goranci, and Peng]{DurfeeGGP19}
David Durfee, Yu~Gao, Gramoz Goranci, and Richard Peng.
\newblock Fully dynamic spectral vertex sparsifiers and applications.
\newblock In \emph{Proc.\ 51st Annual {ACM} {SIGACT} Symposium on Theory of
  Computing, {STOC}}, pages 914--925, 2019.
\newblock \doi{10.1145/3313276.3316379}.

\bibitem[Elkin(2007)]{elkin2007spanners}
Michael Elkin.
\newblock A near-optimal distributed fully dynamic algorithm for maintaining
  sparse spanners.
\newblock In \emph{Proc.\ 26th Annual {ACM} Symposium on Principles of
  Distributed Computing ({PODC})}, pages 185--194, 2007.
\newblock \doi{10.1145/1281100.1281128}.

\bibitem[Feamster and Rexford(2017)]{feamster2017and}
Nick Feamster and Jennifer Rexford.
\newblock Why (and how) networks should run themselves.
\newblock \emph{arXiv preprint arXiv:1710.11583}, 2017.

\bibitem[Feamster et~al.(2018)Feamster, Rexford, and
  Willinger]{DBLP:conf/sigcomm/2018selfdn}
Nick Feamster, Jennifer Rexford, and Walter Willinger, editors.
\newblock \emph{Proceedings of the Afternoon Workshop on Self-Driving Networks,
  SelfDN@SIGCOMM 2018, Budapest, Hungary, August 24, 2018}, 2018. {ACM}.
\newblock URL \url{http://dl.acm.org/citation.cfm?id=3229584}.

\bibitem[Foerster and Schmid(2019)]{DBLP:conf/nca/Foerster019}
Klaus-Tycho Foerster and Stefan Schmid.
\newblock Distributed consistent network updates in {SDNs}: Local verification
  for global guarantees.
\newblock In \emph{18th {IEEE} International Symposium on Network Computing and
  Applications {NCA}}, pages 1--4. {IEEE}, 2019.
\newblock \doi{10.1109/NCA.2019.8935035}.

\bibitem[Foerster et~al.(2017)Foerster, Richter, Seidel, and
  Wattenhofer]{DBLP:conf/icdcn/FoersterRSW17}
Klaus-Tycho Foerster, Oliver Richter, Jochen Seidel, and Roger Wattenhofer.
\newblock Local checkability in dynamic networks.
\newblock In \emph{Proc. of the 18th International Conference on Distributed
  Computing and Networking (ICDCN)}, pages 4:1--10. {ACM}, 2017.
\newblock \doi{10.1145/3007748.3007779}.

\bibitem[Foerster et~al.(2018)Foerster, Luedi, Seidel, and
  Wattenhofer]{DBLP:journals/tcs/FoersterLSW18}
Klaus-Tycho Foerster, Thomas Luedi, Jochen Seidel, and Roger Wattenhofer.
\newblock Local checkability, no strings attached: (a)cyclicity, reachability,
  loop free updates in sdns.
\newblock \emph{Theoretical Computer Science}, 709:\penalty0 48--63, 2018.
\newblock \doi{10.1016/j.tcs.2016.11.018}.

\bibitem[Foerster et~al.(2019{\natexlab{a}})Foerster, Hirvonen, Suomela, and
  Schmid]{supported-local}
Klaus-Tycho Foerster, Juho Hirvonen, Jukka Suomela, and Stefan Schmid.
\newblock On the power of preprocessing in decentralized network optimization.
\newblock In \emph{Proc. {IEEE} Conference on Computer Communications
  ({INFOCOM} 2019)}, 2019{\natexlab{a}}.
\newblock \doi{10.1109/INFOCOM.2019.8737382}.

\bibitem[Foerster et~al.(2019{\natexlab{b}})Foerster, Korhonen, Rybicki, and
  Schmid]{foerster2019preprocessing}
Klaus-Tycho Foerster, Janne~H. Korhonen, Joel Rybicki, and Stefan Schmid.
\newblock Does preprocessing help under congestion?
\newblock In \emph{Proc.\ 38nd {ACM} Symposium on Principles of Distributed
  Computing, ({PODC} 2019)}, pages 259--261, 2019{\natexlab{b}}.
\newblock \doi{10.1145/3293611.3331581}.

\bibitem[Fortz and Thorup(2000)]{fortz2000internet}
Bernard Fortz and Mikkel Thorup.
\newblock Internet traffic engineering by optimizing {OSPF} weights.
\newblock In \emph{Proc. IEEE INFOCOM}, volume~2, pages 519--528. IEEE, 2000.
\newblock \doi{10.1109/INFCOM.2000.832225}.

\bibitem[Frank et~al.(2013)Frank, Poese, Lin, Smaragdakis, Feldmann, Maggs,
  Rake, Uhlig, and Weber]{frank2013pushing}
Benjamin Frank, Ingmar Poese, Yin Lin, Georgios Smaragdakis, Anja Feldmann,
  Bruce Maggs, Jannis Rake, Steve Uhlig, and Rick Weber.
\newblock Pushing {CDN-ISP} collaboration to the limit.
\newblock \emph{ACM SIGCOMM Computer Communication Review}, 43\penalty0
  (3):\penalty0 34--44, 2013.
\newblock \doi{10.1145/2500098.2500103}.

\bibitem[Gilbert and Li(2020)]{GilbertL20}
Seth Gilbert and Lawrence Li.
\newblock How fast can you update your {MST}?
\newblock In \emph{Proc.\ 32nd ACM Symposium on Parallelism in Algorithms and
  Architectures (SPAA)}, pages 531---533, 2020.
\newblock \doi{10.1145/3350755.3400240}.
\newblock URL \url{https://arxiv.org/abs/2002.06762}.

\bibitem[Goranci et~al.(2017)Goranci, Henzinger, and Peng]{GoranciHP17}
Gramoz Goranci, Monika Henzinger, and Pan Peng.
\newblock The power of vertex sparsifiers in dynamic graph algorithms.
\newblock In \emph{25th Annual European Symposium on Algorithms, {ESA}}, pages
  45:1--45:14, 2017.
\newblock \doi{10.4230/LIPIcs.ESA.2017.45}.

\bibitem[Gupta and Khan(2018)]{GuptaK18}
Manoj Gupta and Shahbaz Khan.
\newblock Simple dynamic algorithms for maximal independent set and other
  problems.
\newblock \emph{CoRR}, abs/1804.01823, 2018.
\newblock URL \url{http://arxiv.org/abs/1804.01823}.

\bibitem[Henzinger(2018)]{Henzinger18}
Monika Henzinger.
\newblock The state of the art in dynamic graph algorithms.
\newblock In \emph{{SOFSEM} 2018: Theory and Practice of Computer Science -
  44th International Conference on Current Trends in Theory and Practice of
  Computer Science}, pages 40--44, 2018.
\newblock \doi{10.1007/978-3-319-73117-9\_3}.

\bibitem[Henzinger et~al.(2015)Henzinger, Krinninger, Nanongkai, and
  Saranurak]{HenzingerKNS15}
Monika Henzinger, Sebastian Krinninger, Danupon Nanongkai, and Thatchaphol
  Saranurak.
\newblock Unifying and strengthening hardness for dynamic problems via the
  online matrix-vector multiplication conjecture.
\newblock In \emph{Proc.\ 47th Annual {ACM} on Symposium on Theory of
  Computing, {STOC}}, pages 21--30, 2015.
\newblock \doi{10.1145/2746539.2746609}.

\bibitem[Henzinger et~al.(2016)Henzinger, Krinninger, and
  Nanongkai]{HenzingerKN16}
Monika Henzinger, Sebastian Krinninger, and Danupon Nanongkai.
\newblock Dynamic approximate all-pairs shortest paths: Breaking the {$O(mn)$}
  barrier and derandomization.
\newblock \emph{{SIAM} J. Comput.}, 45\penalty0 (3):\penalty0 947--1006, 2016.
\newblock \doi{10.1137/140957299}.

\bibitem[Henzinger and King(1999)]{HenzingerK99}
Monika~Rauch Henzinger and Valerie King.
\newblock Randomized fully dynamic graph algorithms with polylogarithmic time
  per operation.
\newblock \emph{J. {ACM}}, 46\penalty0 (4):\penalty0 502--516, 1999.
\newblock \doi{10.1145/320211.320215}.

\bibitem[Holm et~al.(2001)Holm, de~Lichtenberg, and Thorup]{HolmLT01}
Jacob Holm, Kristian de~Lichtenberg, and Mikkel Thorup.
\newblock Poly-logarithmic deterministic fully-dynamic algorithms for
  connectivity, minimum spanning tree, 2-edge, and biconnectivity.
\newblock \emph{J. {ACM}}, 48\penalty0 (4):\penalty0 723--760, 2001.
\newblock \doi{10.1145/502090.502095}.

\bibitem[Italiano(1991)]{Italiano91}
Giuseppe~F. Italiano.
\newblock Distributed algorithms for updating shortest paths.
\newblock In \emph{Proc.\ 5th International Workshop on Distributed Algorithms
  ({WDAG})}, pages 200--211, 1991.
\newblock \doi{10.1007/BFb0022448}.

\bibitem[Italiano et~al.(2019)Italiano, Lattanzi, Mirrokni, and
  Parotsidis]{italiano2019dynamic}
Giuseppe~F. Italiano, Silvio Lattanzi, Vahab~S. Mirrokni, and Nikos Parotsidis.
\newblock Dynamic algorithms for the massively parallel computation model.
\newblock In \emph{Proc.\ 31st ACM Symposium on Parallelism in Algorithms and
  Architectures (SPAA 2019)}, pages 49--58, New York, NY, USA, 2019.
  Association for Computing Machinery.
\newblock \doi{10.1145/3323165.3323202}.

\bibitem[K{\"{o}}nig and Wattenhofer(2013)]{KonigW13}
Michael K{\"{o}}nig and Roger Wattenhofer.
\newblock On local fixing.
\newblock In \emph{Proc.\ 17th International Conference on Principles of
  Distributed Systems ({OPODIS} 2013)}, pages 191--205, 2013.
\newblock \doi{10.1007/978-3-319-03850-6\_14}.

\bibitem[Korhonen and Rybicki(2017)]{korhonen2017deterministic}
Janne~H. Korhonen and Joel Rybicki.
\newblock Deterministic subgraph detection in broadcast {CONGEST}.
\newblock In \emph{Proc.\ 21st International Conference on Principles of
  Distributed Systems ({OPODIS} 2017)}, 2017.
\newblock \doi{10.4230/LIPIcs.OPODIS.2017.4}.

\bibitem[Kuhn et~al.(2010)Kuhn, Lynch, and Oshman]{kuhn2010distributed}
Fabian Kuhn, Nancy~A. Lynch, and Rotem Oshman.
\newblock Distributed computation in dynamic networks.
\newblock In \emph{Proc.\ 42nd {ACM} Symposium on Theory of Computing, {STOC}},
  pages 513--522, 2010.
\newblock \doi{10.1145/1806689.1806760}.

\bibitem[Kushilevitz and Nisan(1997)]{KushilevitzN97}
Eyal Kushilevitz and Noam Nisan.
\newblock \emph{Communication complexity}.
\newblock Cambridge University Press, 1997.

\bibitem[Lenzen(2013)]{lenzen2013optimal}
Christoph Lenzen.
\newblock Optimal deterministic routing and sorting on the congested clique.
\newblock In \emph{Proc.\ 2013 ACM symposium on Principles of distributed
  computing (PODC 2013)}, pages 42--50, 2013.
\newblock \doi{10.1145/2484239.2501983}.

\bibitem[Lotker et~al.(2005)Lotker, Patt{-}Shamir, Pavlov, and
  Peleg]{lotker2005mst}
Zvi Lotker, Boaz Patt{-}Shamir, Elan Pavlov, and David Peleg.
\newblock Minimum-weight spanning tree construction in {$O(\log \log n)$}
  communication rounds.
\newblock \emph{{SIAM} Journal on Computing}, 35\penalty0 (1):\penalty0
  120--131, 2005.
\newblock \doi{10.1137/S0097539704441848}.

\bibitem[Michel and Keller(2017)]{DBLP:conf/sds/MichelK17}
Oliver Michel and Eric Keller.
\newblock {SDN} in wide-area networks: {A} survey.
\newblock In \emph{Proc.\ 4th International Conference on Software Defined
  Systems, ({SDS} 2017)}, pages 37--42. {IEEE}, 2017.
\newblock \doi{10.1109/SDS.2017.7939138}.

\bibitem[Neiman and Solomon(2016)]{NeimanS16}
Ofer Neiman and Shay Solomon.
\newblock Simple deterministic algorithms for fully dynamic maximal matching.
\newblock \emph{{ACM} Trans. Algorithms}, 12\penalty0 (1):\penalty0 7:1--7:15,
  2016.
\newblock \doi{10.1145/2700206}.

\bibitem[Networks(2020)]{juniper-driving}
Juniper Networks.
\newblock Expel complexity with a self-driving network, 2020.
\newblock URL \url{https://www.juniper.net/us/en/dm/the-self-driving-network/}.

\bibitem[Nowicki(2019)]{nowicki2019deterministic}
Krzysztof Nowicki.
\newblock A deterministic algorithm for the {MST} problem in constant rounds of
  congested clique, 2019.
\newblock URL \url{http://arxiv.org/abs/1912.04239}.
\newblock arXiv:1912.04239 [cs.DS].

\bibitem[Nowicki and Onak(2020)]{NowickiO20}
Krzysztof Nowicki and Krzysztof Onak.
\newblock Dynamic graph algorithms with batch updates in the massively parallel
  computation model.
\newblock \emph{CoRR}, abs/2002.07800, 2020.
\newblock URL \url{https://arxiv.org/abs/2002.07800}.

\bibitem[O'Dell and Wattenhofer(2005)]{DBLP:conf/dialm/ODellW05}
Regina O'Dell and Roger Wattenhofer.
\newblock Information dissemination in highly dynamic graphs.
\newblock In Suman Banerjee and Samrat Ganguly, editors, \emph{Proc.\
  {DIALM-POMC} Joint Workshop on Foundations of Mobile Computing}, pages
  104--110, 2005.
\newblock \doi{10.1145/1080810.1080828}.

\bibitem[Parter et~al.(2016)Parter, Peleg, and Solomon]{parter2016local}
Merav Parter, David Peleg, and Shay Solomon.
\newblock Local-on-average distributed tasks.
\newblock In \emph{Proc.\ 27th Annual {ACM-SIAM} Symposium on Discrete
  Algorithms ({SODA})}, pages 220--239, 2016.
\newblock \doi{10.1137/1.9781611974331.ch17}.

\bibitem[Peleg(1998)]{peleg1998distributed}
David Peleg.
\newblock Distributed matroid basis completion via elimination upcast and
  distributed correction of minimum-weight spanning trees.
\newblock In \emph{Proc.\ International Colloquium on Automata, Languages, and
  Programming (ICALP 1998)}, pages 164--175. Springer, 1998.
\newblock \doi{10.1007/BFb0055050}.

\bibitem[Peleg(2000)]{Peleg00}
David Peleg.
\newblock \emph{Distributed Computing: A Locality-Sensitive Approach}.
\newblock Monographs on Discrete Mathematics and Applications. Society for
  Industrial and Applied Mathematics, 2000.
\newblock ISBN 9780898714647.

\bibitem[Perlman(1985)]{DBLP:conf/sigcomm/Perlman85}
Radia~J. Perlman.
\newblock An algorithm for distributed computation of a spanning tree in an
  extended {LAN}.
\newblock In William Lidinsky and Bart~W. Stuck, editors, \emph{Proc.\ 9th
  Symposium on Data Communications ({SIGCOMM})}, pages 44--53. {ACM}, 1985.
\newblock \doi{10.1145/319056.319004}.

\bibitem[Peterson and Davie(2011)]{PetersonD11}
Larry~L. Peterson and Bruce~S. Davie.
\newblock \emph{Computer Networks, Fifth Edition: A Systems Approach}.
\newblock Morgan Kaufmann Publishers Inc., San Francisco, CA, USA, 5th edition,
  2011.
\newblock ISBN 0123850592.

\bibitem[Rauch~Henzinger and Thorup(1996)]{rauch1996improved}
Monika Rauch~Henzinger and Mikkel Thorup.
\newblock Improved sampling with applications to dynamic graph algorithms.
\newblock In \emph{Proc.\ International Colloquium on Automata, Languages, and
  Programming (ICALP 1998)}, pages 290--299. Springer, 1996.

\bibitem[Razborov(1992)]{Razborov92}
Alexander~A. Razborov.
\newblock On the distributional complexity of disjointness.
\newblock \emph{Theor. Comput. Sci.}, 106\penalty0 (2):\penalty0 385--390,
  1992.
\newblock \doi{10.1016/0304-3975(92)90260-M}.

\bibitem[Schmid and Suomela(2013)]{schmid2013exploiting}
Stefan Schmid and Jukka Suomela.
\newblock Exploiting locality in distributed {SDN} control.
\newblock In \emph{Proc.\ 2nd ACM SIGCOMM Workshop on Hot Topics in Software
  Defined Networking (HotSDN 2013)}, pages 121--126. ACM Press, 2013.
\newblock \doi{10.1145/2491185.2491198}.

\bibitem[Schrijver(2003)]{schrijver-book}
A.~Schrijver.
\newblock \emph{Combinatorial Optimization - Polyhedra and Efficiency}.
\newblock Springer, 2003.

\bibitem[Simsiri et~al.(2016)Simsiri, Tangwongsan, Tirthapura, and
  Wu]{simsiri2016work}
Natcha Simsiri, Kanat Tangwongsan, Srikanta Tirthapura, and Kun-Lung Wu.
\newblock Work-efficient parallel union-find with applications to incremental
  graph connectivity.
\newblock In \emph{European Conference on Parallel Processing}, pages 561--573.
  Springer, 2016.
\newblock \doi{10.1007/978-3-319-43659-3\_41}.

\bibitem[Tseng et~al.(2019)Tseng, Dhulipala, and Blelloch]{tseng2019batch}
Thomas Tseng, Laxman Dhulipala, and Guy Blelloch.
\newblock Batch-parallel euler tour trees.
\newblock In \emph{Proc.\ 21st Meeting on Algorithm Engineering and Experiments
  ({ALENEX} 2019)}, pages 92--106, 2019.
\newblock \doi{10.1137/1.9781611975499.8}.

\bibitem[Wang(2020)]{hu-driving}
David Wang.
\newblock Moving towards autonomous driving networks, 2020.
\newblock URL
  \url{https://www.huawei.com/en/publications/communicate/87/moving-towards-autonomous-driving-networks}.

\bibitem[{Wang} et~al.(2015){Wang}, {Wu}, and {Chou}]{sdn-stp}
S.~{Wang}, C.~{Wu}, and C.~{Chou}.
\newblock Constructing an optimal spanning tree over a hybrid network with
  {SDN} and legacy switches.
\newblock In \emph{2015 IEEE Symposium on Computers and Communication (ISCC)},
  pages 502--507, 2015.
\newblock \doi{10.1109/ISCC.2015.7405564}.

\end{thebibliography}

\end{document}